\documentclass[11pt]{article}

\usepackage{enumitem}
\usepackage{amsfonts}
\usepackage{amsmath}
\usepackage{amsthm}
\usepackage[dvipsnames]{xcolor}
\usepackage[a4paper, margin=1in]{geometry}
\usepackage{booktabs} % For formal tables
\usepackage[ruled]{algorithm2e} % For algorithms

\usepackage{natbib}
\usepackage{authblk}

\usepackage{hyperref}

\newcommand{\bigparen}[1]{\left( #1 \right)}
\newcommand{\bigbraces}[1]{\left\{ #1 \right\}}

\newcommand{\realnum}{\mathbb{R}}

\newcommand{\argmax}{\operatornamewithlimits{arg\,max}}

\setlist[itemize]{noitemsep, topsep=0pt}

\newtheorem{theorem}{Theorem}
\newtheorem{lemma}{Lemma}
\newtheorem{corollary}{Corollary}[theorem]
\newtheorem{proposition}{Proposition}
\theoremstyle{definition}
\newtheorem{definition}{Definition}
\theoremstyle{definition}
\newtheorem{example}{Example}

\title{Price Competition in Linear Fisher Markets: \\Stability, Equilibrium and Personalization}

\date{}

\author{Juncheng Li}
\author{Pingzhong Tang}
\affil{Tsinghua University}

\allowdisplaybreaks

\begin{document}
	
\maketitle

\begin{abstract}
	Linear Fisher market is one of the most fundamental economic models.
%	We study price competition among sellers in linear Fisher markets, one of the most fundamental economic models.
	The market is traditionally examined on the basis of individual's price-taking behavior.
	However, this assumption breaks in markets such as online advertising and e-commerce, where several oligopolists dominate the market and are  able to compete with each other via strategic actions.
	Motivated by this, we study the price competition among sellers in linear Fisher markets.
	From an algorithmic game-theoretic perspective,
	we establish a model to analyze  behaviors of buyers and sellers that are driven by utility-maximizing purposes and also constrained by computational tractability.
	The main economic observation is the role played by personalization: the classic benchmark market outcome, namely competitive equilibrium, remains to be a steady-state if every buyer must be treated ``equally''; however, sellers have the incentive to personalize, and as a result the market would become more unpredictable and less efficient.
	In addition, we build a series of algorithmic and complexity results along the road to justify our modeling choices and reveal market structures.
	We find interesting connections between our model and other computational problems such as stable matching, network flow, etc.
	We believe these results and techniques are of independent interest.
\end{abstract}	

\section{Introduction}

Fisher market (and the more general Arrow-Debreu market \citep{arrow1954existence,mckenzie1954equilibrium}) is one of the most foundational economic models.
It serves as an excellent first approximation to those real-world markets that consist of a \textit{large} number of buyers and sellers such that no individual has the power to effectively influence prices.
In addition, due to the many nice properties of the market solution, named \textit{competitive equilibrium}, it is also widely implemented in markets where there is only a \textit{single} benevolent seller/owner who wants to allocate items to buyers/recipients in a fair and envy-free way \citep{foley1967resource,varian1974equity}.
The model and its idea keeps receiving attention in the literature as it is applied and adapted to various situations, such as fair division of indivisible items \citep{budish2011combinatorial,cole2015approximating,cole2017convex,babaioff2019fair} or a continuum of items \citep{gao2021infinite}, recommender systems \citep{kroer2019scalable,kroer2019computing,kroer2023fair} and online advertising auction markets \citep{conitzer2022pacing,conitzer2022multiplicative}.

In spite of the broad applicability of the model and the appeal of competitive equilibrium, the price-taking assumption fails to hold in markets like oligopolies, where the items are controlled by a group of sellers who are powerful enough to set prices arbitrarily for items of its own.
In these industries, the market outcome is not only driven by the omnipotent power of overall supply and demand, but also the revenue-seeking behaviors of individual sellers.
%\footnote{Though oligopoly is a well-suited example, the focus of this paper is the impact of the \textit{competition} rather than the \textit{number} of sellers.}
Seller competition is an ongoing topic in the literature of game theory, from the classic models of competition on the basis of prices and production quantities (see, e.g., Chapter 6 of the book by  \citet{dixon2017surfing}), to the more complex ones like the model by \citet{selcuk2017auctions} and \citet{paes2020competitive} where sellers compete for buyers' money through selecting suitable trading mechanisms.
We cannot review the vast literature exhaustively here, but to the best of our knowledge, no existing work puts seller competition in the context of (or close to) the fundamental Fisher market model.
In this paper, we initiate the study of this problem and approach it from both \textit{game-theoretic} and \textit{computational} perspectives.

% computation and dynamics
%tâtonnement, trading post, proportional dynamics
% cce, off-equilibrium dynamics
%complexity

Starting with the work by \citet{megiddo1991total}, theoretical computer scientists began to notice an important yet previously overlooked aspect of economic models -- the computation of solution concepts.
If an equilibrium cannot be computed efficiently, it would be hard, if not impossible, to implement for the market owner and less likely to emerge as the outcome of social learning.
%On the contrary, a solution concept will be significantly strengthened if it is computationally tractable, preferably with some natural and decentralized dynamics that leads to it. 
In a linear Fisher market, the competitive equilibrium is not only poly-time solvable  \citep{eisenberg1959consensus,orlin2010improved,shmyrev2009algorithm}, but could also be reached through well-motivated fast-converging dynamics \citep{cole2008fast,zhang2011proportional}.
It is therefore of great computational soundness.
On the other hand, if some equilibrium concept is shown to be hard to achieve, those economically weaker but computationally tractable solutions would become more favorable.\footnote{One such example is \textit{coarse correlated equilibrium}, which can be reached through decentralized no-regret learning in a large variety of games where Nash equilibrium is hard to compute.}
In this paper, we propose a new price competition model over linear Fisher markets, and characterize the set of market outcomes that are both economically stable and computationally feasible.
% economically; personalization
Our model incorporates several features absent in traditional formulations, and in particular the role of buyer \textit{personalization} is examined.
The market is more well-behaved without personalization, but sellers have the incentive to do so, and the resulting market might become more unpredictable and less efficient.
% computational structure
Besides revealing market structures, our computational results also demonstrate some interesting relationships between our model and problems such as stable matching and network flow.
These new connections and techniques may be of independent interest as well.

\subsection{Markets Driven by Price Competition.}
A linear Fisher market consists of a set of divisible items and a set of budget-constrained buyers with linear utilities.
There is a long line of research on Fisher markets that sits at the intersection of theoretical computer science and economics.
Our work is of great importance to the understanding of the price competition among sellers in linear Fisher markets.
One of the real-world examples is online advertising: for instance, YouTube Shorts, Instagram Reels and TikTok all feature short videos and compete for very similar audiences.
In theory, researchers also show that online advertising markets can be well captured by the linear Fisher market model \citep{conitzer2022pacing,conitzer2022multiplicative} (but they study a single advertising platform rather multiple competing ones).
Below we will introduce the key features brought by the price-setting power of sellers, and the rationale behind them.

In Fisher markets,
buyers are traditionally assumed to be \textit{supply-unaware} in the sense that they demand and only demand items with the maximum \textit{bang-per-bucks} (the value acquired for each unit of money spent on the item) regardless of supply.
In real-world markets, however, people seldom let their money sit idle.
So in our model we let trade happen whenever possible, even if the bang-per-buck is not optimal.

In a price-taking market, the role played by sellers is negligible, since a change of item ownership never affects the equilibrium outcome.
%In both t\^{a}tonnement and proportional dynamics, items operate independently and each is implicitly assumed to belong to a different seller.
In our model, besides the power to set prices, sellers could make decisions based on the aggregated behaviors of buyers on all the items in their possession.
Consequently, item-ownership becomes influential.

% price personalization
Traditionally there is no buyer personalization in Fisher markets, which is a natural result of information symmetry and perfect competition.
%Even if there is personalization, it is more coarse-grained and appears in the form of, e.g., quantity discounts or limited-time offers.
With the advent of big data, however, sellers are now able to literally personalize for every individual.
E.g., in online advertising, sellers (such as Google and Meta) make the definitive decisions on pricing and allocation. Trade happens in real time and no arbitrage is possible. The items sold (opportunities to place an ad) are themselves personalized.
This motivates us to observe the implication of personalization and it turns out to be of great consequence.

As the main focus of this paper is the competition among sellers,
%buyers are put in a subordinate position and assumed to behave in a simpler way.
the complex strategic behaviors of buyers are not considered, except that each buyer must be served a bundle of items optimal in hindsight.\footnote{Alternative models might include, e.g., the allowance of arbitrage or buyers forming coalitions, which are interesting directions for future works.}
Previously researchers \citep{adsul2010nash,branzei2014fisher} studied a game model on the opposite side, where buyers can \textit{misreport} their valuations and the market owner will allocate items via the competitive equilibrium given the reported valuations. This is hard to achieve in markets where it is the seller who has the information advantage.
For instance, in online advertising, valuations are predicted based on the data and machine learning model in the control of the seller, who would typically hide most of the detailed information from buyers for both privacy and profit considerations.
%We do not consider valuation misreport in this paper.
%This modeling choice is made since the focus of our paper is on markets mainly driven by the competition among sellers.
%A real-world example is the online advertising market, where Google and Meta are the dominant players and advertisers are relatively less powerful.

\subsection{Our Results}

We propose a model over linear Fisher market with additionally a set of sellers, each owning a subset of items of the underlying market.
Each seller's strategy consists of two parts: \textit{prices} and \textit{priorities}.
Prices are allowed to be personalized, and the priority is a total preorder dictating who the seller prefers to sell an item to (independent of prices).\footnote{
	As an illustration, consider e-commerce platforms. Though each item comes with a posted price,  it can be personalized through coupons, and priorities can be implemented (implicitly) by the recommendation algorithm.}
A myopic seller may use a \textit{natural} priority that fully agrees with the pricing, which represents a greedy revenue-seeking behavior.
Possibly surprisingly, we will show that a natural priority may be suboptimal, so arbitrary priorities constitute a necessary strategic component in general.
We do not incorporate complex buyer behaviors in our model; the only requirement is that they should be served an optimal bundle \textit{in hindsight}.

%In this model, sellers compete with each other by setting prices and \textit{priorities} over buyers for each item.
%A priority is a total preorder dictating who the seller prefers to sell an item to.
%In most cases we will consider \textit{natural} priorities that agree with the pricing (we will soon see why unnatural ones are necessary nonetheless).
%If personalization is forbidden, the price should be uniform across buyers and there is no priority.
%Given any pricing, each buyer has a preference over items based on bang-per-bucks and it is a knapsack problem to maximize its valuation.
%We will focus on the (pure) Nash equilibrium of the game since it would be the limit point of any convergent dynamics.
%To make the normal form game well-defined, we should specify the utility function for sellers, preferably corresponding to a trade outcome (or a set of outcomes) that is the most plausible to happen, given any action (prices and priorities) profile of sellers.

Given sellers' prices and priorities, we impose that the trade outcome (also referred to as the allocation) should satisfy \textit{stability}, which dictates that no pair of a buyer and a seller has the incentive to deviate and form a new trade such that both are made strictly happier (w.r.t. valuation for buyer and priorities for seller).
Besides its intuitive economic implication, we show that a stable allocation always exists and is strongly poly-time computable.
It is therefore justified as a requisite of the trade outcome.
A stable allocation is generally not unique, which can be interpreted as an inherent uncertainty of the market and is interesting in its own right.

We proceed to analyze different categories of markets characterized by the degree to which the market is personalized.

\textbf{Non-personalized pricing with natural priorities.}
We first consider the case where no personalization is allowed, and find an interesting ``one-sided'' tractability result: it is strongly poly-time solvable to find the (market-wise) revenue-maximizing stable allocation, while it is APX-hard to minimize it.
Moreover, the set of revenue-optimal\footnote{In this paper, ``revenue-optimal'' is used interchangeably with ``revenue-maximizing'' as it is typically used.} stable allocations coincides perfectly with the set of seller-wise Pareto-optimal ones.
Therefore we further impose that the trade outcome should satisfy revenue-optimality.
Nonetheless, the issue of multiplicity remains.
As a result, we adopt a maximin approach and study the \textit{maximin equilibrium} where no seller could \textit{guarantee} itself a higher revenue after deviation regardless of the allocation selection.
We show that the competitive equilibrium always forms a maximin equilibrium, but is generally not the unique one.
We also explore properties of those non-competitive equilibria via case study, which include market non-clearance, ownership awareness, and what kind of sellers may gain or lose in the competition.

\textbf{Non-personalized pricing with arbitrary priorities.}
Things become complicated even when just the priorities are personalized.\footnote{In reality, sometimes the pricing is expected to be non-personalized, while a personalization of priority is possibly easier to emerge.
	For instance, in ride-hailing, the pricing depends only on buyer-oblivious features like the driver/vehicle and the route, but the platform (such as Uber and Didi) controls the assignment of drivers to passengers, which creates the kind of personalized priorities we study here.}
It is now APX-hard to find either the revenue-maximizing or the revenue-minimizing stable allocation, which indicates that the market outcome is relatively unpredictable.
We also give an example where multiplicity is not an issue but there is a profitable deviation for a seller at the competitive equilibrium.
This shows the incentive of a seller to personalize even just the priorities, and the inefficiency it might bring.

\textbf{Personalized pricing.}
Interestingly, with personalized pricing, we identify a natural subset of prices and priorities that could lead to a \textit{unique} stable allocation, as long as there is no tie involved (and priorities should be natural).
However, the APX-hardness remains in general, which implies that different tie-breaking choices may lead to significant market changes.
Furthermore, we give an example where, given any prices and priorities, there always exists a profitable deviation for some seller.
These evidence shows the market is highly unpredictable.

\textbf{Competitive pricing for duopoly.}
Delving deeper into the aforementioned example, we find that it is partly due to the static nature of our  model that no equilibrium exists.
Hence we deviate a little from the main model to incorporate some dynamic feature, i.e., sellers are now allowed to retaliate the opponent by matching the deviated prices (more precisely the bang-per-bucks).
The model is defined only for duopoly as the price-matching behavior is more meaningful with two sellers.
In general, the market outcome in this new model is still unpredictable, so we focus on a specific class of markets that share the same underlying Fisher market with the no-equilibrium example while having different item ownership.
We show that the competitive equilibrium will form a stationary structure when the market is relatively \textit{balanced}, while less efficient outcomes will appear as sellers become more different from each other.

\subsection{Computational Techniques and Structures}

In the main text we will focus on the analysis of the market model.
Here we give an overview of the techniques used in and the insights gained from our computational results,
and refer readers to the appendix for the details.

The backbone of the algorithm computing a stable allocation (Theorem \ref{thm:gale_shapley} and Appendix \ref{app:gale_shapley}) is adapted from the classic Gale-Shapley algorithm \citep{gale1962college}: here buyers propose to items in decreasing order of preferences, and items reject buyers with lower priorities.
However, a naive adaption to the fractional case will result in excessively long or even infinite loops.
% not sure whether an example is needed
Much efforts are thus dedicated to find these loops and resolve them efficiently.
The definition of stable allocation actually bears some resemblance to stable matching, but neither of these two concepts contains the other as a special case: in our model the preference is connected with money, while in stable matching there is no concept of budgets.

An interesting new connection is established between our non-personalized  market and network flow.
A revenue-optimal stable allocation (with arbitrary pricing) can be encoded in a flow network, and computed with a particular implementation of the Edmonds-Karp algorithm \citep{edmonds1972theoretical}
such that certain economic property is maintained for  every intermediate flow throughout the execution (Theorem \ref{thm:one_sided_tractability} and Appendix \ref{app:network_flow_and_maximal_matching}).
Some proofs (including those for personalized markets) can be viewed intuitively as examining how the flow/allocation is dragged towards a certain direction by the gravity of stability.
Prior to our work, network flow has been used in Fisher markets, albeit in ways quite different from ours.
\citet{vegh2012strongly} and \citet{vegh2014concave} encode the computation of competitive equilibrium in linear Fisher markets as a single-shot generalized minimum-cost flow and generalized maximum flow problem, respectively.
\citet{devanur2008market} compute the competitive equilibrium by solving a series of maximum flows corresponding to different price vectors.
In their model, buyers only demand items they prefer the most and items will never become under-demanded during the execution of the algorithm.

The ``one-sided'' tractability of revenue optimization in non-personalized pricing markets resembles the case for maximal/minimal matching, and we  show that their computational structures are indeed closely related (Lemma \ref{lemma:minimum_maximal_matching} in Appendix \ref{app:network_flow_and_maximal_matching}).

Most of our APX-hardness results on revenue maximization/minimization also imply APX-completeness, since the revenues of two stable allocations differ by at most a factor of 2 (which is tight; see Appendix \ref{app:2approximation}) and a stable allocation is poly-time computable.

\section{Preliminaries}

A \textbf{Fisher market} consists of: (1) a set $J$ of $m$ divisible items, each with a supply (or quantity) normalized to 1; (2) a set $I$ of $n$ buyers; (3) a budget $B_i$ for each buyer $i$; (4) a utility function $u_i: [0, 1]^{n \times m} \rightarrow \realnum$ for each buyer $i$, which takes the allocation of the whole market $x \in [0, 1]^{n \times m}$ as the input (budgets have no value for buyers). In this paper, there is no externality and $u_i$ depends only on the allocation received by buyer $i$ itself.
%	 \footnote{For convenience of notation, we still let $u_i$ take the market allocation as the input.}
Note that traditionally there is no explicit mention of seller or item-ownership.

A Fisher market is \textbf{linear}, if $u_i$ has the form
$
u_i(x) =
\sum_j v_{i, j} x_{i, j},
$
where $x_{i, j}$ is the fraction of item $j$ allocated to buyer $i$, and $v_{i, j} \geq 0$ is the value of item $j$ to buyer $i$.
For each item, there is at least one buyer $i$ with $v_{i, j} > 0$.
Each item will come with a price $p_j$, and buyers will try to acquire as much utility as possible within their budget constraints.

A \textbf{competitive equilibrium} $(p^*, x^*)$ of a linear Fisher market consists of a pricing vector (or simply pricing or prices) $p^* \in [0, +\infty)^m$ and an allocation $x^* \in [0, 1]^{n \times m}$ satisfying:
\begin{itemize}
	\item Buyer optimization or envy-freeness:
	$
		x_i^* \in \argmax_{x_i \in [0, +\infty)^{m}: \sum_j p^*_j x_{i, j} \leq B_i} \sum_j v_{i, j} x_{i, j};
	$
	\item Market clearing:
	$
		\sum_i x_{i, j}^* = 1, \forall j.
	$
\end{itemize}

Later we will study equilibrium concepts other than the competitive equilibrium. For convenience, we will refer to the allocation/pricing at the corresponding equilibrium as the \textit{equilibrium allocation/pricing} if its meaning is unambiguous from the context. Otherwise we will explicitly associate the quantity with the equilibrium name, e.g., competitive equilibrium allocation/pricing.

Below we state some well-known results that will be frequently used:
\begin{itemize}
	%	\item $x$ is an equilibrium allocation if and only if it solves the \textbf{Eisenberg-Gale (EG) program}, which can be written as follows:
	%	
	%	\begin{align*}
		%		\max \quad & \sum_i B_i \ln \bigparen{u_i}
		%		\\
		%		\text{s.t.} \quad &
		%		u_i \leq \sum_j v_{i, j} x_{i, j}, \forall i;
		%		\\
		%		&
		%		\sum_i x_{i, j} \leq 1, \forall j;
		%		\\
		%		& x_{i, j} \geq 0, \forall i, j.
		%	\end{align*}
	\item An equilibrium always exists.
	Equilibrium pricing $p^*$ is unique.
	There may be multiple equilibrium allocations, but buyers' utilities are identical across them.
	\item Buyers always deplete their budgets at equilibrium.
	\item Given prices $p$, the \textbf{bang-per-buck} of item $j$ for buyer $i$ is $\frac{v_{i, j}}{p_j}$, the value acquired for each unit of money spent on buying item $j$.
	At equilibrium, every buyer only gets items with the \textit{highest} bang-per-buck, i.e., if
	$
		x_{i, j} > 0$,
		then
		$
		\frac{v_{i, j}}{p_j}
		\geq
		\frac{v_{i, j'}}{p_{j'}}, \forall j'.
	$
	Note that this requirement will be relaxed in our price competition model.
\end{itemize}

\section{Fixed Pricing Markets and Stability}
\label{sec:stability}

\begin{definition} \label{def:fixed_pricing_game}
	A \textbf{fixed pricing market} is defined over a linear Fisher market with a set of sellers, of which each seller $k$ owns a disjoint subset $J_k$ of items.
	Trade proceeds with two stages:
	
	(1, pricing and prioritizing) Each seller $k$ sets a \textbf{personalized price} $p_{i, j}$ for each buyer $i$ and each item $j \in J_k$.
	Besides, each seller also chooses a total preorder $\geq_j$, called \textbf{priority}, over buyers for each item $j \in J_k$.
	A priority is \textbf{natural w.r.t. prices $p$}, if $i \geq_j i'$ whenever $p_{i, j} \geq p_{i', j}$.
	When we focus on \textbf{non-personalized} pricing, we mean the price (not priority) of an item should be equal across buyers.
	Given the pricing, call the total preorder $\geq_i$ that satisfies $j \geq_i j'$ whenever $\frac{v_{i, j}}{p_{i, j}} \geq \frac{v_{i, j'}}{p_{i, j'}}$ the \textbf{preference} of buyer $i$ over items.
	We will also say an item \emph{prioritizes} some buyer or a buyer \emph{prefers} some item (over the other).
	
	(2, trading) Trade happens according to some allocation $x$, where $x_{i, j} \in [0, 1]$ is the fraction of item $j$ allocated to buyer $i$, and $\sum_i x_{i, j} \leq 1, \forall j$.
	The utility of each buyer is the total valuation it acquires, and the revenue of each seller is the total payment it receives.
\end{definition}

% explain the rationale of the whole definition
% fixed pricing market is the main object we deal with in this paper
% note that fow now it only incorporates economic considerations
% explain the rationale or practical meaning of priority

Fixed pricing markets set the stage for the analysis of seller competition in this paper.
The definition is mostly straightforward: sellers post prices and trade happens.\footnote{Our model does not only capture markets with posted prices. E.g., in online advertising markets with second price auction and auto-bidders, prices are determined in real time by both buyer competition (via second price) and the platform (via instruments such as boosts and reserve prices; see, e.g., the work by \citet{balseiro2021robust}). In addition, auto-bidders using the multiplicative pacing strategy behave in precisely the same way as in our model, i.e., they buy items in decreasing order of bang-per-bucks until depleting budgets \citep{conitzer2022multiplicative}.}
A natural priority represents the seller's \textit{greedy} revenue-optimizing behavior to sell an item with a higher price.
Preference represents a similar rationale on the buyer side.
Later we will show that a natural priority may be myopic and sub-optimal in the long run.
Therefore it is economically significant to allow arbitrary priority.

In Definition \ref{def:fixed_pricing_game}, we do not impose any condition on the trade outcome $x$.
Below we will introduce a minimal requirement on $x$ (given pricing $p$), and show that it is both economically and computationally reasonable.

\begin{definition}
	Given an allocation $x$, let $r_i = \min_{j: x_{i, j} > 0} \frac{v_{i, j}}{p_{i, j}}$ be the minimum bang-per-buck among items bought by buyer $i$.
	$x$ is \textbf{stable} for buyer $i$, if it satisfies:  (1) budget-feasible: $\sum_{j} x_{i, j} p_{i, j} \leq B_i$; (2) if $x_{i, j} > 0$, then $v_{i, j} > 0$; (3) if $x_{i, j} < 1$, then either (3.a) the bang-per-buck of $j$ is no larger than $r_i$ and $i$'s budget is already depleted, or (3.b) item $j$ is fully allocated and $i' \geq_j i$ for every other buyer $i'$ with $x_{i', j} > 0$.
	$x$ is stable if it is stable for all buyers.
	
	Given prices, priorities and an allocation, a fraction of item $j$ is said to be \textbf{available} to buyer $i$, if it is either unsold, bought by buyer $i$ itself, or bought by some buyer $i'$ with $i >_j i'$ (i.e., $i \geq_j i'$ and $i' \not \geq_j i$).
	An item is available to buyer $i$, if a positive fraction of it is available to buyer $i$.
\end{definition}

From a single buyer's perspective, given its available items, the optimal bundle is the solution of a fractional knapsack problem, i.e., buying available items in decreasing order of bang-per-bucks until depleting its budget or buying all available items.
Stability ensures that each buyer get its optimal bundle in hindsight.
As for sellers, if an item is available to buyer $i$, its owner will happily sell (a fraction of) it to buyer $i$ (by withholding the item from buyers with strictly lower priorities, if necessary).
If an allocation is not stable, there would be a pair of buyer and seller who would like to deviate from it and make themselves strictly happier with respect to preference and priority, respectively.
Therefore stability is a necessary economic requirement.
Furthermore, a stable allocation always exists and is poly-time computable.

% it is tempted to consider a normal form game over the market

\begin{theorem}[Existence and computation of stable allocation] \label{thm:gale_shapley}
	A stable allocation always exists and can be computed in strongly poly-time.
	If all inputs are rational numbers, there exists a poly-time computable stable allocation where all entries are also rational numbers.
\end{theorem}

The algorithm and proof is given in Appendix \ref{app:gale_shapley}.
% multiplicity
Stability does not always lead to a unique allocation.
Nonetheless, the multiplicity captures an uncertain aspect of the market, into which we will probe further in later sections.

\begin{definition}
	Given prices $p$, an allocation $x$ is called \textbf{compatible} for buyer $i$, if $x$ is budget-feasible and each item $j$ with $x_{i, j} > 0$ has a bang-per-buck that is weakly larger than the bang-per-buck of any item that is currently available to it.
	In other words, buyer $i$ gets its optimal bundle over available items with its currently spent money.
	Allocation $x$ is compatible, if it is compatible for all buyers.
\end{definition}

\begin{lemma} \label{lemma:compatible_stable}
	A compatible allocation is stable, if and only if, for each buyer $i$, either $B_i = \sum_j x_{i, j} p_{i, j}$, or it gets all items available to it.
\end{lemma}

Compatibility characterizes a transient ``stable'' state by focusing only on the budgets already spent.
It will be useful in the market analysis.

%The correctness comes from the fact that the compatibility of the current allocation is always maintained after each round of the main while loop (from line \ref{alg:gale_shapley:main_loop} to \ref{alg:gale_shapley:end}).
%The algorithm is efficient since each round of the main while loop leads to some event that can only happen a polynomial number of times.

%We conclude this section by bounding the difference of revenues between any two stable allocations.

\section{Non-personalized Pricing Markets with Natural Priorities}
\label{sec:non_personalized_pricing_game}

In this section, we study markets that are \textit{fully non-personalized}: every buyer is treated equally with respect to both pricing and priority.
A fraction of item $j$ is available to buyer $i$ if and only if it is either not allocated, or owned by buyer $i$ itself.
Instability only occurs when there exist a buyer $i_0$ and an item $j_0$ such that $v_{i_0, j_0} > 0, \sum_i x_{i, j_0} < 1$ and it holds that either (1) buyer $i_0$'s remaining budget is not zero or (2) buyer $i_0$ strictly prefers $j_0$ to some other item $j_1$ with $x_{i_0, j_1} > 0$.
Missing proofs in this section are given in Appendix \ref{app:non_personalized_pricing_game}.

\subsection{Seller-side Optimality and Equilibrium Concept}

% example of revenue sub-optimal
\begin{center}
	\begin{tabular}{ccc}
		\toprule
		valuation & item 1 (seller 1) &  item 2 (seller 2) \\
		\midrule
		buyer 1 (budget 1) & 1 & 1  \\
		buyer 2 (budget 1) & 0 & 1  \\
		\bottomrule
	\end{tabular}
\end{center}

We begin with the example depicted in the above table.
If $p_1 = p_2 = 1$, it is easy to check that $x_{1, 2} = 1, x_{1, 1} = x_{2, 1} = x_{2, 2} = 0$ forms a stable allocation.
However, it is seller-wise (weakly) Pareto-dominated in revenue by another stable allocation: $y_{1, 1} = y_{2, 2} = 1, y_{1, 2} = y_{2, 1} = 0$.
This motivates us to wonder whether allocations such as $x$ are less likely to occur in general.

To formalize and justify the above intuition, we have the following complexity result and Pareto-efficiency guarantee.\footnote{Corollary \ref{corollary:compatible_extension} is categorized as a corollary since it is a direct consequence of the \textit{proof} of Theorem \ref{thm:one_sided_tractability}.}

\begin{theorem}[One-sided tractability of revenue optimization] \label{thm:one_sided_tractability}
	A stable allocation with maximum revenue can be computed in strongly poly-time.
	On the other side, it is APX-hard to compute the minimum revenue over all stable allocations.
\end{theorem}

\begin{corollary}[Seller-side Pareto-optimality] \label{corollary:compatible_extension}
	For any compatible allocation $x$, there is a revenue-optimal stable allocation $y$ item-wise (thus seller-wise) covering $x$, i.e., $\sum_i x_{i, j} \leq \sum_i y_{i, j}, \forall j \in J$.
	Since stability implies compatibility, any revenue-optimal stable allocation is seller-wise Pareto-optimal (in allocation and revenue).
\end{corollary}

% define equilibrium concept directly and explain the rationale
% admit that it is a weak concept, but it captures many interesting outcomes
Corollary \ref{corollary:compatible_extension} shows that the set of market-wise revenue-optimal stable allocations is exactly the set of seller-side Pareto-optimal ones.
The contrast between the hardness of minimizing revenue and the tractability of maximizing revenue further supports imposing revenue-optimality as a market requisite.
Therefore it should be reasonable to restrict our attention to revenue-optimal stable allocation in fully non-personalized markets.
%Combined with the one-sided tractability, it should be reasonable to restrict our attention to revenue-optimal stable allocation in fully non-personalized markets.
Nonetheless, revenue-optimality is still unable to produce a unique allocation for any pricing, which hinders us from formulating a normal form game of price competition.
To circumvent it, we adopt a maximin approach and define the following equilibrium concept to capture those states that we believe are the most likely to emerge as the market outcome.

\begin{definition}
	Given pricing $p$,
	let $R_k(p)$ be the minimum revenue received by seller $k$ over all market-wise revenue-optimal stable allocations with pricing $p$.
	In a non-personalized pricing market with natural priorities,
	a \textbf{maximin equilibrium} $(p, x)$ consists of a pricing $p$ and an allocation $x$ satisfying: (1) $x$ is a revenue-optimal stable allocation with $p$; (2) for any seller $k$ and any pricing $p'$ that differs from $p$ only on items owned by $k$, $R_k(p') \leq \sum_{i \in I, j \in J_k} p_{j} x_{i, j}$.
\end{definition}

% explanation
At a non-maximin-equilibrium state, there will be a deviation for some seller that \textit{guarantees} itself a strictly higher revenue.
Our equilibrium concept eliminates these clearly unsteady outcomes.
Though in reality a seller might attempt a possibly profitable (but non-guaranteed) deviation, the maximin equilibrium is economically meaningful and it includes many interesting market states as will be detailed below.

\subsection{Equilibrium Analysis}

By Corollary \ref{corollary:compatible_extension}, any compatible allocation has a revenue-optimal stable allocation that item-wise covers it. To show a deviation cannot improve the deviator's worst-case revenue, one sufficient condition is to find a compatible allocation where buyers already spend too much on other sellers, ensuring the deviator cannot be better off even if all the unspent money goes to it.
The above idea will be repeated used in this section to establish maximin equilibrium.

\begin{theorem} \label{thm:non_personalized_nash}
	Any competitive equilibrium $(p^*, x^*)$ is also a maximin equilibrium.
\end{theorem}

\begin{proof}
	Consider any unilateral deviated price vector $p$ that differs from $p^*$ for a single seller $k$.
	Let $x^*$ be a competitive equilibrium allocation.
	Let $H$ be the set of items with $p_j > p^*_j$.
	Construct $x$ as follows:
	\begin{displaymath}
		x_{i, j} = \left\{
		\begin{array}{ll}
			x^*_{i, j}, & \text{if } j \notin H; \\
			0, & \text{if } j \in H.
		\end{array}
		\right.
	\end{displaymath}
	By the property of competitive equilibrium, items outside $H$ are fully allocated in $x$, and if $x^*_{i, j} > 0$, then  $\frac{p^*_{j'}}{v_{i, j'}} \geq \frac{p^*_j}{v_{i, j}}, \forall j'$.
	As a result, if $x_{i, j} > 0$, then $j \notin H$ and $\frac{p_{j'}}{v_{i, j'}} >  \frac{p^*_{j'}}{v_{i, j'}} \geq \frac{p^*_j}{v_{i, j}} \geq \frac{p'_j}{v_{i, j}}, \forall j' \in H$.
	Therefore $x$ is compatible, and by corollary \ref{corollary:compatible_extension}, there exists a revenue-optimal stable allocation $y$ such that $\sum_i x_{i, j} \leq \sum_i y_{i, j}, \forall j \in J$.
	This implies that items are fully sold in $y$ for all sellers other than $k$, so $k$'s worst-case revenue at $p$ cannot exceed its competitive equilibrium revenue.
\end{proof}

Theorem \ref{thm:non_personalized_nash} shows that, even if sellers have the price-setting power, as long as no personalization is allowed, the competitive equilibrium remains to be a steady-state.
This strengthens the foundation of this classic equilibrium concept from a game-theoretic perspective of seller competition.
Also note that the competitive equilibrium is ownership-oblivious, while in general the set of maximin equilibria is not a singleton and it changes as the ownership changes.

% irrelavant of item ownership

%The set of maximin pricing equilibria is generally not a singleton.
%The market may not even clear at equilibrium.
To explore further properties of the equilibria, consider the following representative example.

\begin{example} \label{example:non_personalized}
	The market consists of 2 buyers and 2 sellers.
	Each seller has a budget of 2, and the valuations and item-ownership are given in the following table.
	\begin{center}
		\begin{tabular}{cccc}
			\toprule
			& seller 1 & \multicolumn{2}{c}{seller 2} \\
			& item 1 & item 2 & item 3 \\
			\midrule
			buyer 1 & 2 & 1 & 0 \\
			buyer 2 & 1/3 & 1 & 1 \\
			\bottomrule
		\end{tabular}
	\end{center}
	At the unique competitive equilibrium, $p^*_1 = 2, p^*_2 = p^*_3 = 1$, buyer 1 gets item 1, and buyer 2 gets item 2 and 3.
\end{example}

\textbf{Multiplicity of maximin equilibria.}
There is another maximin equilibrium with pricing: $p_1 = \frac{4}{3}, p_2 = \frac{2}{3}, p_3 = 2$.
There is a unique revenue-optimal stable allocation, where buyer 1 gets item 1 and 2, while buyer 2 gets item 3.
To see it is a maximin equilibrium, consider all the possible deviations as follows:
\begin{itemize}\itemsep0em
	\item Seller 2 strictly increases $p_2$.
	Then allocation $x$ with $x_{1, 1} = 1$ and $x_{i, j} = 0, \forall (i, j) \neq (1, 1)$ is compatible, so seller 2 cannot increase worst-case revenue.
	\item $p_2$ does not increase.
	Allocation $x$ with $x_{1, 1} = x_{1, 2} = 1$ and 0 otherwise is compatible.
	Note that we do not need to consider deviation of $p_3$.
	\item Seller 1 increases $p_1$.
	Allocation $x$ with $x_{1,2}=x_{2,3}=1$ and 0 otherwise is compatible.
	\item Seller 1 at most earns the price of item 1, so decreasing $p_1$ cannot increase revenue.
\end{itemize}

\textbf{Market non-clearance.}
From the above analysis we can see the price of item 3 will not bring extra incentive issues for both sellers, as long as it remains more attractive than item 1 for buyer 2.
Therefore raising $p_3$ to 3 (with buyer 2 still depleting its budget on item 3) also leads to an equilibrium.
In this case, item 3 is not fully sold.

\textbf{Importance of overall competitiveness.}
Intuitively, the above market non-clearing maximin equilibrium exists due to the fact that seller 1 is too weak to be competitive against seller 2 w.r.t. buyer 2 (even if $v_{2, 1}$ is not very small compared to $v_{2, 2} + v_{2, 3}$).
Actually $4/3$ is the minimum revenue seller 1 would get from any maximin equilibrium.
It can only be improved if $v_{2, 1}$ is raised to at least $2/3$ (with other valuations unchanged).
In contrast, seller 2 is competitive enough w.r.t. buyer 1, so it can earn from the price competition.
$v_{1, 2}$ is actually \textit{tight} in terms of the best revenue seller 2 could get at maximin equilibrium.
In comparison, the competitive equilibrium remains the same if we raised $v_{2, 1}$ to 1 or decreased $v_{1, 2}$ to 0.
% chip war?

\textbf{Ownership awareness.}
Suppose item 3 is now owned by seller 1.
Then  the pricing $p_1 = \frac{4}{3}, p_2 = \frac{2}{3}, p_3 = 2$ ceases to be a maximin equilibrium.
In this case, seller 2 could raise $p_2$ to 1 such that item 2 could remain as the most preferred item of buyer 2 and would always be fully sold at maximin equilibrium.

%\paragraph{Budget surplus.}

% risk averse
% competitive nature
% uncertainty of the market state
In summary, within a fully non-personalized market, there is generally some inherent uncertainty that cannot be resolved by just pricing and competition (e.g., due to user habits, brand awareness, etc.).
The competitive equilibrium still forms a stable structure from a risk-averse point of view, though other (possibly less efficient) outcomes are also likely to emerge as a steady-state.

\section{Non-personalized Pricing Markets with Arbitrary Priorities}
\label{sec:non_personalized_arbitrary}

In this section, we consider the personalization of just priorities.
Missing proofs in this section can be found in Appendix \ref{app:non_personalized_arbitrary}.
We will first show that, though it seems to be a minor modification, the market outcome immediately becomes more unpredictable as arbitrary priorities are introduced.

\begin{theorem} \label{thm:apx_non_personalized_arbitrary}
	In non-personalized pricing markets with arbitrary priorities,
	for fixed pricing and priorities, the problem of finding the maximum revenue over stable allocations is APX-hard.
	Furthermore, for fixed pricing, it is APX-hard to find the maximum revenue over all stable allocations and all priorities. 
\end{theorem}

Theorem \ref{thm:apx_non_personalized_arbitrary} demonstrates two layers of intractability.
First, combing with Theorem \ref{thm:one_sided_tractability}, we have that both maximizing and minimizing the market-wise revenue is intractable for fixed strategies of sellers.
This makes it evident that the market outcome is relatively unpredictable.
In addition, in markets like ride-hailing where prices change less frequently while priorities are more flexible, the second half of Theorem \ref{thm:apx_non_personalized_arbitrary} shows that it is also hard to find the optimal priorities for both the  whole market and individual sellers.

% competitive equilibrium ceases to be equilibrium
Moreover, even if we take the conservative maximin approach, the competitive equilibrium may cease to be a stationary state.
To see this, consider Example \ref{example:non_personalized} again.
Recall that, at competitive equilibrium, $p_1^* = 2$, $p_2^* = p_3^* = 1$, buyer 1 gets item 1 and buyer 2 gets item 2 and 3.
Priorities do not affect allocation here.
Suppose seller 2 deviates to $p_2 = 2/3$, $p_3 = 2$ and strictly prioritizes buyer 1 over buyer 2 for item 2.
Since buyer 1 now strictly prefers item 2 over item 1, it will buy item 2 in whole and deplete the rest of budget on item 1.
Buyer 2 also prefers item 2 the most, but it is not prioritized to buy.
So it could only buy item 3 at a higher price (which is still more attractive then the now unsold item 1).
As a result, seller 2 can guarantee itself a higher revenue by deviating from the competitive equilibrium pricing.
Note that in this example no allocation selection is involved.

% ?  some inefficient equilibrium can be eliminated as well

% equilibrium existence and when competitie equilibrium remains equilibrium is unknown

\section{Personalized Pricing Markets}
\label{sec:personalized_pricing_game}

The situation becomes even more complex with price personalization.
We start with an interesting result that identifies a natural situation where the allocation is unique and thus exempted from the selection problem.
Missing proofs in this section are given in Appendix \ref{app:personalized_pricing_game}.

\begin{theorem} \label{thm:unique}
	If items' priorities $\geq_j$ and buyers' preferences over items $\geq_i$ are all strict, and $\geq_j$ is natural, then there is a unique stable allocation.
\end{theorem}

However, even with natural priorities, the market becomes unpredictable when ties exist.

\begin{theorem} \label{thm:apx_complete}
	The problem of finding the stable allocation with the maximum revenue is APX-hard, even if the priorities are natural.
	The APX-hardness of minimizing revenue is implied by Theorem \ref{thm:one_sided_tractability}.
	%	Combined with theorem \ref{thm:gale_shapley} and \ref{thm:2_approximation}, the problem is APX-complete.
\end{theorem}

Note that ties can be easily broken by small changes, e.g., subsidies/coupons in e-commerce, boosts in auto-bidding, or an update of the machine learning model that matches buyers and items.
Hence the hardness result indicates that these perturbations might bring a large overhaul of the trading outcome.
The proof also shows that a natural priority may not be optimal, even if the personalized prices are all distinct from one another.

To make things worse, a market may admit no stationary state, regardless of the selection criterion, as shown by the following result.

\begin{theorem} \label{proposition:no_pne}
	Consider a market consisting of 2 buyers and 2 sellers.
	The valuations and item-ownership are given in Table \ref{tab:no_pne}.
	\begin{table}[h]
		\centering
		\begin{tabular}{ccc}
			\toprule
			valuation & item 1 (seller 1) & item 2 (seller 2)
			\\ \midrule
			buyer 1 (budget 2) & $1$ & $2$ 
			\\
			buyer 2 (budget 2) & $2$ & $1$ 
			\\ \bottomrule
		\end{tabular}
			\caption{A fixed pricing market without stationary state.}
		\label{tab:no_pne}
	\end{table}
	With any pricing of seller 2, seller 1 always has a response guaranteeing a revenue strictly larger than 2, regardless of allocation selection.
\end{theorem}
\begin{proof}
	We analyze seller 1's best responses for different combinations of $p_{1, 2}$ and $p_{2, 2}$.
	\begin{itemize}
		\item $p_{1, 2} > 0, p_{2, 2} > 1$.
		Consider $p_{1, 1} = \frac{1}{4} p_{1, 2}$, and $p_{2, 1} = 1 + p_{2, 2}$.
		Since $p_{1, 1} < \frac{1}{2} p_{1, 2}$ and $p_{2, 1} < 2p_{1, 2}$, both buyers strictly prefer item 1, and thus for any stable allocation, either both buyers spend all budgets on item 1, or item 1 is fully sold.
		For the former case, seller 1's revenue is 4.
		Below assume item 1 is fully sold.
		\begin{itemize}
			\item If $p_{1, 1} \geq p_{2, 1} > 2$, seller 1's revenue is at least $\min(p_{1, 1}, p_{2, 1}) > 2$.
			\item If $p_{1, 1} < p_{2, 1}$, then seller 1 prefers buyer 2 and buyer 2 will spend all its budget on item 1.
			But buyer 2 cannot buy item 1 in whole, so a fraction of item 1 will be sold to buyer 1, and seller 1's total revenue is also strictly larger than 2.
		\end{itemize}
		\item $p_{1, 2} < 2,  p_{2, 2} \leq 1$.
		Then seller 2's revenue is strictly less than 2.
		Seller 1 can accommodate all the rest budgets simply by choosing $p_{1, 1} = p_{2, 1} = 4$.
		\item $p_{1, 2} > 4,  p_{2, 2} \leq 1$.
		Choose $p_{1,1} = 2 + \epsilon$ for some small $\epsilon > 0$ such that $p_{1,1} < \frac{p_{1,2}}{2}$ and buyer 1 strictly prefers item 1 to item 2. Then buyer 1 will deplete its budget on item 1 but cannot buy it entirely, and buyer 2 will spend a positive amount on the remainder of item 1 since item 2 costs at most 1 dollar.
%		Again choose $p_{1, 1} = p_{2, 1} = 4$.
%		Buyer 1 will spend all its budget on item 1, and buyer 2 will spend at least 1 dollar on item 1.
		\item $p_{1, 2} \in (2, 4], p_{2, 2} \leq 1$.
		Consider $p_{1, 1} = \frac{1 - \epsilon}{2} \cdot p_{1, 2}$ and $p_{2, 1} = \frac{2}{\epsilon}$.
		By theorem \ref{thm:unique}, the stable allocation is unique.
		Buyer 2 spends at most 1 dollar on item 2, so it spends at least 1 dollar on item 1 (it has higher priority) but buys at most a fraction of $\epsilon$.
		Buyer 1 strictly prefers item 1, so the rest part of item 1 (excluding the fraction sold to buyer 2) must be sold to buyer 1, which generates a revenue of at least $ \frac{(1 - \epsilon)^2}{2} \cdot p_{1, 2}$.
		With sufficiently small $\epsilon$, we have $ \frac{(1 - \epsilon)^2}{2} \cdot p_{1, 2} + 1 > 2$.
		\item $p_{1, 2} = 2, p_{2, 2} \leq 1$. Choose $p_{1, 1} = \epsilon$ and $p_{2, 1} = \frac{1}{\epsilon}$ for sufficiently small $\epsilon$.
		By theorem \ref{thm:unique}, the stable allocation is unique.
		Buyer 1 strictly prefers item 1, buyer 2 strictly prefers item 2, but $p_{1, 1}$ and $p_{2, 2}$ are both smaller than 2, so both items will be fully sold in stable allocation.
		On the other hand, buyer 1 has top priority on item 2, buyer 2 has top priority on item 1, and $p_{1, 2}$ and $p_{2, 1}$ are both no less than 2, so both buyers will deplete their budgets in stable allocation.
		Furthermore, buyer 1 must spend some money on item 1 since item 1 cannot be fully sold to buyer 2.
		Therefore buyer 1 spends strictly less than 2 on item 2, and the revenue of seller 2 is also strictly less than 2.
		Since both buyers will deplete budgets, seller 1's revenue will be strictly larger than 2.
	\end{itemize}
\end{proof}

By symmetry, seller 2 always has a response guaranteeing a revenue strictly larger than 2, as well.
Clearly the revenues of both sellers cannot exceed 2 simultaneously.
Therefore there exists no equilibrium no matter how it is defined.
From the proof we can also see how the prices can be personalized to benefit the seller.

\subsection{Competitive Duopoly Pricing Markets}
\label{subsec:competitive_duopoly_pricing_markets}

One contributor to the success of the deviation in Theorem \ref{proposition:no_pne} is that the opponent's pricing is always fixed.
Consider the seemingly promising strategy for buyer 2: $p_{1, 2} = 2, p_{2, 2} \leq 1$ (which is similar to the competitive equilibrium pricing).
After seller 1 pushes the bang-per-buck ratio for buyer 1 high enough, some money of buyer 2 will be attracted from item 2 to item 1, and thus seller 2's revenue will decrease.
However, seller 1 can always deplete buyer 2's remaining budget (after buying item 2) by pushing the bang-per-buck ratio for buyer 2 low enough.

Suppose, however, that seller 1 does carry out the deviation.
Seller 2 might not sit still, and there is a straightforward (not necessarily optimal) countermeasure: lowering the bang-per-buck ratio for buyer 2 to a level that is slightly higher than seller 1.
This will not change buyer preference, and seller 2 can at least earn some more money from buyer 1.

Motivated by the above observation, we attempt to get some glimpse into the possible steady states of a personalized market by incorporating a ``competitive'' feature into the model, such that buyer's preference is still based on each seller's reported prices, but trading is settled with the \textit{lowest} bang-per-buck ratio among all sellers.
Note that this can be approximated arbitrarily closely as described above.
To make it amenable, we also introduce \textit{bang-per-buck pricing}\footnote{It relates to auto-bidding in first price auctions markets \citep{conitzer2022pacing}, and it provides a weak form of optimality for the seller. See Appendix \ref{app:bang_per_buck_pricing} for detailed discussion.} and only define this modified market for two sellers.

%To circumvent all these issues and make analysis tractable, we define the \textbf{competitive duopoly pricing game} and focus on the study of examples that are well-behaved yet representative to some degree.
%With our newly defined game, Nash equilibrium now exists and it will depend on \textit{item ownership}.
%Competitive equilibrium may not be reached and the Nash equilibrium allocation may even be not Pareto-optimal for buyers.

\begin{definition}
	A \textbf{competitive duopoly pricing market} is defined over a linear Fisher market  with two sellers.
	The trade proceeds with three stages:
	
	(1, reporting bang-per-buck)
	Seller 1 and 2 each sets a bang-per-buck $a_i \geq 0$ and $b_i \geq 0$ for each buyer $i$.
	
	(2, pricing and prioritizing)
	Buyer $i$'s preference depends on which one of $a_i$ and $b_i$ is smaller.
	Buyers have no preference over items owned by the same seller since the bang-per-bucks are equalized.
	Item $j$'s price for buyer $i$ is instead given by $p_{i, j} = v_{i, j} \cdot \max(a_i, b_i)$.
	Priorities should be natural w.r.t. $p$.
	
	(3, trading) Trade happens according to an allocation $x$ that is stable w.r.t. the priorities and preferences defined above (instead of the clearing prices, which give the same bang-per-buck for all items).
\end{definition}

Note that the market is still unpredictable in general, since the construction used in the proof of Theorem \ref{thm:apx_complete} implies that the problem is APX-hard even for the restricted case where (1) item's priorities are natural and strict among buyers with non-zero valuations; (2) items can be grouped into two sets, and every buyer is indifferent among items in the same group. %(in the reduction, choice and clause-literal items constitute one group, buffer items constitute the other).
%The result also transfers to the problem for each seller.

Below we do not bother with the allocation selection problem.
Instead we only focus on a generalized version of the market in Table \ref{tab:no_pne} where the selection rule will not be an issue.
Then we could directly analyze the Nash equilibrium of the normal form game, i.e., the bang-per-buck profile $(a, b)$ such that each seller is best responding to its opponent.

\begin{table}[h]
	\centering
	\begin{tabular}{ccccc}
		\toprule
		& \multicolumn{2}{c}{seller 1} & \multicolumn{2}{c}{seller 2}
		\\
		& item 1  & item 2 & item 3 & item 4
		\\ \midrule
		buyer 1 & $s$ & $2t$ & $2(1 - t)$ & $1 - s$
		\\
		buyer 2 & $2s$ & $t$ & $1 - t$ & $2(1 - s)$
		\\ \bottomrule
	\end{tabular}
	\caption{The underlying market is identical to the one in Table \ref{tab:no_pne}, but the item ownership is different.}
	\label{tab:starting_example}
\end{table}

\begin{example} \label{example:duopoly}
	Consider  the market depicted in Table \ref{tab:starting_example} where each buyer has a budget of 2.
	The underlying Fisher market is essentially identical to the one in Table \ref{tab:no_pne}: one item, denoted as $j_1$, that buyer 1 values 2 dollars and buyer 2 values 1 dollar, and the other item $j_2$ that buyer 1 values 1 dollar and buyer 2 values 2 dollars.
	The difference lies in the ownership: seller 1 now owns a fraction $t$ of $j_1$ and a fraction $s$ of item $j_2$, while seller 2 owns the rest.
\end{example}

\begin{proposition} \label{proposition:duopoly_example_ce}
	The competitive equilibrium bang-per-buck profile $a_1 = a_2 = b_1 = b_2 = 1$ forms a Nash equilibrium if and only if $s, t \in [1/3, 2/3]$.
\end{proposition}

Intuitively, with $s, t \in [1/3, 2/3]$, the sellers are kind of \textit{similar} to each other in both the volume and the quality of the items they provide.
This makes the market competitive enough to reach an efficient state.

\begin{proposition} \label{proposition:non_pareto_ne}
	With $s = 1$ and $t = 0$ (the ownership in Table \ref{tab:no_pne}),
	$a_2 = b_1 = \frac{4}{3}, a_1 = b_2 = 1$ forms a Nash equilibrium.
	Both buyers will get half of both items at the Nash equilibrium.
\end{proposition}

In this case, though with similar market sizes, each seller provides a distinctive product that attracts a different audience.
Though the market clears, the allocation is not Pareto-optimal for buyers.
In this particular case, seller competition hurts buyers.
Sellers seem to get nothing from personalization, but if one of them did not do so, the opponent will gain the advantage.

\section{Conclusion and Future Works}

% price competition, linear fisher markets, online advertising
% stability (see thesis)
% personalization
In this paper, we examine the price competition among sellers within the foundational model of linear Fisher market.
We show that, given sellers’ pricing (and priorities), the set of allocations satisfying the “stability” property constitutes the steady-states of the market. Stable allocations are both economically reasonable and computationally practical, but not unique, which represents an uncertain aspect of the market.
We move on to studying markets with different degrees of personalization, which turn out to have quite distinct properties.
If no buyer personalization is allowed, the competitive equilibrium remains to form a steady-state (the maximin equilibrium), albeit not the unique one.
On the other hand, personalization could bring sellers more revenue, but the resulting market becomes more unpredictable, as evidenced by the two-sided hardness of maximizing/minimizing revenue and the non-existence of (static) equilibrium.

There are several interesting directions for future research.
Most of our theorems (except for Theorem \ref{proposition:no_pne}) hold regardless of the number of sellers, but the properties of various solution concepts are mainly examined via small instances containing two sellers.
Though they provide some meaningful insights, it is important to explore these properties in a more general sense.
In Section \ref{sec:personalized_pricing_game} we give much evidence to show that personalized markets are unpredictable.
It is of great importance to explore what kinds of (possibly dynamic) states these markets might end up with (Section \ref{subsec:competitive_duopoly_pricing_markets} serves as a preliminary attempt in this direction).

\bibliographystyle{abbrvnat}
\bibliography{bibliography}

\clearpage
\newpage

\appendix

\section{Proof of Theorem \ref{thm:gale_shapley}}
\label{app:gale_shapley}

Theorem \ref{thm:gale_shapley} is a direct corollary of the following lemma.
\begin{lemma}
	Algorithm \ref{alg:gale_shapley} returns a stable allocation for any (rational) input parameters in strongly poly-time.
\end{lemma}

\begin{algorithm*}
		\small
	\caption{Fractional Gale-Shapley algorithm}
	\label{alg:gale_shapley}
	\LinesNumbered
	\DontPrintSemicolon
	\SetAlgoNoEnd
	\SetKwProg{Fn}{Function}{:}{}
	\SetKwFunction{FNormalUpdate}{normal\_update}
	
	\KwIn{valuations $v_{i, j}$, prices $p_{i, j}$, priorities $\geq_j$, budgets $B_i$}
	\nllabel{alg:gale_shapley:start}
	for each item $j$, maintain temporary allocations $x_{i, j}$ (initially $ 0$) and the set of temporary buyers $\{i: x_{i, j} > 0\}$ \;
	for each buyer $i$, maintain its remaining budget (initially $ B_i$) and the set of currently proposable items (initially all items with $v_{i, j} > 0$) \;
	for each buyer $i$, sort items in descending order of bang-per-bucks; buyers could only propose to buy the item with the maximum bang-per-buck among currently proposable items (ties are broken in an arbitrary but fixed way) \;
	\While{there exists buyer $i$ with non-zero remaining budget and an item $j$ that is proposable to $i$}{
		\nllabel{alg:gale_shapley:main_loop}
		\FNormalUpdate{$i$, $j$} \;
		\nllabel{alg:gale_shapley:naive_end}
		
		\nllabel{alg:gale_shapley:no_loop_1}
		\If{(1) item $j$ is fully allocated before the last \FNormalUpdate, and (2) no buyer leaves the temporary buyer set of item $j$ during the last \FNormalUpdate}{
			\nllabel{alg:gale_shapley:find_circle}
			\Repeat{either condition (1) or (2) is violated for the last \FNormalUpdate, or $i'$ has no proposable item, or a proposal pair $(i_0, j_0)$ is encountered three times within this repeat-until  loop}{
				assertion: when condition (1) and (2) are both satisfied after one call of \FNormalUpdate{$i$, $j$}, except for buyer $i$, there should be only one buyer, denoted $i'$, whose allocation is changed (transferred to buyer $i$) during the update \nllabel{alg:gale_shapley:unique_successor} \;
				\If{there exists item $j'$ that is proposable to buyer $i'$ (whose remaining budget must be non-zero due to the last \FNormalUpdate)}{
					\FNormalUpdate{$i'$, $j'$}
				}
			}
			
			\If{the above repeat-until loop exits due to a proposal pair is encountered three times}{
				\nllabel{alg:gale_shapley:no_loop_2}
				denote the proposal sequence between the last two \texttt{normal\_update($i_0, j_0$)} as $(i_1, j_1), \dots, (i_\ell, j_\ell)$  \;
				assertion: the proposal sequence between the first two \texttt{normal\_update($i_0, j_0$)} is also $(i_1, j_1), \dots, (i_\ell, j_\ell)$ \nllabel{alg:gale_shapley:assertion_start} \;
				assertion: $i_0, i_1, \dots, i_\ell$ and $j_0, j_1, \dots, j_\ell$ contain no duplicate respectively
				\;
				assertion: $i_1, \dots, i_\ell$ has no remaining budget before the second \texttt{normal\_update($i_0, j_0$)} \nllabel{alg:gale_shapley:assertion_end} \;
				
				Denote $i_0$'s budget before the second and third \texttt{normal\_update($i_0, j_0$)}  as $B_0$ and $B_0 - \Delta B_0$ respectively; 
				let the allocation before the second \texttt{normal\_update($i_0, j_0$)} be $y$;
				during the most recent proposal sequence $(i_1, j_1), \dots, (i_\ell, j_\ell)$, for each $k$, $i_k$ gains $\Delta y_{i_{k + 1}, j_k}$ more of item $j_k$ from $i_{k + 1}  $($i_{\ell + 1} := i_0$), and loses $\Delta y_{i_k, j_{k - 1}}$ of item $j_{k - 1}$ ($j_{0 - 1} := j_\ell$) to $i_{k - 1}$ \;
				if $\Delta B_0 > 0$, let $r = \min \bigbraces{\frac{B_0}{\Delta B_0}, \frac{y_{i_0, j_{\ell}}}{\Delta y_{i_0, j_\ell}}, \dots, \frac{y_{i_\ell, j_{\ell - 1}}}{\Delta y_{i_\ell, j_{\ell - 1}}}}$; otherwise  let $r = \min \bigbraces{\frac{y_{i_0, j_{\ell}}}{\Delta y_{i_0, j_\ell}}, \dots, \frac{y_{i_\ell, j_{\ell - 1}}}{\Delta y_{i_\ell, j_{\ell - 1}}}}$ \;
				update $x_{i_k, j_k} = y_{i_k, j_k} + r \Delta y_{i_{k + 1}, j_k}$, $x_{i_k, j_{k - 1}} = y_{i_k, j_{k - 1}} - r \Delta y_{i_k, j_{k - 1}}$ for each $k \in \{0, \dots, \ell\}$, update budgets and proposability accordingly \;
				\nllabel{alg:gale_shapley:budget_replacement}
			}
		}
		\nllabel{alg:gale_shapley:end}
	}
	
	\nllabel{alg:gale_shapley:normal_update}
	\Fn{\FNormalUpdate{$i$, $j$}}{
		$i$'s demand on $j$ is $d_{i, j} = \min\{1, b_i / p_{i, j} + x_{i, j}\}$; every other temporary buyer $i'$ has a demand $d_{i', j} = x_{i', j}$ \;
		reallocate $j$ by filling demands one-by-one in decreasing order of $\geq_j$ (ties are broken in favor of buyers who come to item $j$ earlier) until no supply is available \;
		if $j$ is currently fully allocated, update proposability: let $i^*$ be the buyer with lowest priority among current buyers of $j$; mark item $j$ as not proposable to each buyer $i'$ with $i^* \geq_j i'$ (including $i^*$ itself) \;
		budgets are updated accordingly \;
	}
\end{algorithm*}

\begin{proof}
	\textbf{Correctness.}
	The algorithm terminates when the current allocation satisfies the condition in lemma \ref{lemma:compatible_stable} (line \ref{alg:gale_shapley:main_loop} of algorithm \ref{alg:gale_shapley}).
	Thus we only need to show that, after each round of the main while loop (from line \ref{alg:gale_shapley:main_loop} to \ref{alg:gale_shapley:end}), the allocation is always compatible.
	
	A single call of \texttt{normal\_update} consists of one buyer proposing to  one item in analogy to the original Gale-Shapley algorithm.
	Using the standard argument for Gale-Shapley algorithm, it is easy to show that \texttt{normal\_update} will maintain the compatibility of the current allocation.
	
	A single round of the main while loop can only finish at either line \ref{alg:gale_shapley:no_loop_1}, \ref{alg:gale_shapley:no_loop_2} or \ref{alg:gale_shapley:budget_replacement}.
	For the first two cases, the algorithm simply conducts a finite number of \texttt{normal\_update}, so the allocation must be compatible.
	
	For the third case, after each \texttt{normal\_update($i, j$)} before line \ref{alg:gale_shapley:no_loop_2}, condition (1) and (2) defined in line \ref{alg:gale_shapley:find_circle} are both satisfied, which means that there is some buyer $i'$ who transfers part of its allocation of item $j$ to buyer $i$ (recall that demands are fulfilled with a fixed order; this also justifies the assertion at line \ref{alg:gale_shapley:unique_successor}).
	Buyer $i$ can buy item $j$ from buyer $i'$, so it has a strictly higher priority than $i'$.
	But buyer $i'$ does not leave the temporary buyer set (condition 2), which can only happen when buyer $i$ runs out of budget.
	This also indicates that, if buyer $i$ gets some money later (due to currently bought items transferred to buyers with higher priorities), the next item it will propose to is still $j$.
	With the above property, we can justify the three assertions from line \ref{alg:gale_shapley:assertion_start} to \ref{alg:gale_shapley:assertion_end}: the first holds since the successor is unique and unchanged; the second holds for if a buyer is duplicated, we should find a shorter circle, and if an item $j_k$ is duplicated, the buyer $i_{k + 1}$ that loses item $j_k$ will lead to a duplicated buyer; the third holds because they already run out of budgets during the proposal sequence between the first and second \texttt{normal\_update($i_0, j_0$)}.
	
	What line \ref{alg:gale_shapley:budget_replacement} does is to roll back to the allocation $y$ before the second \texttt{normal\_update($i_0, j_0$)} (which is compatible as argued above) , and then conduct a simultaneous update for all buyers $i_0, i_1, \dots, i_\ell$ involved in the circle.
	Let the final allocation be $x$.
	Now we only need to prove that $x$ is non-negative, budget-feasible and compatible.
	\begin{enumerate}
		\item (non-negativity) $r \leq \frac{y_{i_k, j_{k - 1}}}{\Delta y_{i_k, j_{k - 1}}}, \forall k$, so $x_{i_k, j_{k - 1}} = y_{i_k, j_{k - 1}} - r \Delta y_{i_k, j_{k - 1}} \geq 0$.
		\item (budget-feasibility)
		For $k \in \{1, \dots, \ell\}$, before the second \texttt{normal\_update($i_0, j_0$)}, $i_k$ has no remaining budget.
		During the last proposal sequence $(i_1, j_1), \dots, (i_\ell, j_\ell)$, $i_k$ gets $p_{i_k, j_{k - 1}} \Delta y_{i_k, j_{k - 1}}$ dollars returned from \texttt{normal\_update($i_{k - 1}, j_{k - 1}$)}, and then spends all of them to buy $\Delta y_{i_{k + 1}, j_{k}}$ of item $j_k$ from $i_{k + 1}$.
		Therefore we have
		\begin{displaymath}
			p_{i_k, j_{k - 1}} \Delta y_{i_k, j_{k - 1}}
			=
			p_{i_k, j_{k}} \Delta y_{i_{k + 1}, j_{k}}.
		\end{displaymath}
		During the simultaneous update, compared to $y$, $i_k$ gets $r \Delta y_{i_{k + 1}, j_k}$ more of item $j_k$, and loses $r \Delta y_{i_k, j_{k - 1}}$ of item $j_{k - 1}$.
		From the above equation, we can see that $i_k$'s budget remains zero.
		
		The situation for $i_0$ is a little different.
		During the last proposal sequence $(i_0, j_0), (i_1, j_1), \dots, (i_\ell, j_\ell)$ (excluding the third \texttt{normal\_update($i_0, j_0$)}), it loses $\Delta y_{i_0, j_\ell}$ of item $j_\ell$, gets $\Delta y_{i_1, j_0}$ more of item $j_0$, and loses $\Delta B_0$ dollars.
		If $\Delta B_0 \leq 0$, after the simultaneous update, its budget will increase and thus feasible.
		Otherwise, since $r \leq \frac{B_0}{\Delta B_0}$, after the simultaneous update, its budget becomes $B_0 - r \Delta B_0 \geq 0$.
		\item (compatibility)
		First note that, after the simultaneous update, each item transfers part of itself from a buyer with strictly lower priority to a buyer with strictly higher priority.
		So if item $j$ is not proposable to buyer $i$ with allocation $y$, it remains not proposable  to $i$ with allocation $x$.
		This indicates that $x$ remains compatible for buyers not involved in the simultaneous update.
		For each $k$, buyer $i_k$ loses part of item $j_{k - 1}$, which is already not proposable to it (since a buyer could only lose an item if it is the buyer with lowest priority among temporary buyers), and gains part of item $j_k$.
		Since $y$ is compatible, and item $j_k$ has the maximum bang-per-buck among available items with allocation $y$, it holds that $x$ is compatible for each buyer $i_k$.
	\end{enumerate}
	
	\textbf{Poly-time termination.}
	A circle $(i_0, j_0), (i_1, j_1), \dots, (i_\ell, j_\ell)$ can be identified within $O(n)$ calls of \texttt{normal\_update}, and each line of algorithm \ref{alg:gale_shapley} can be finished in polynomial time, so we only need to show that the main while loop can only be repeated a polynomial number of rounds.
	Recall that one round of the main while loop finishes with any one of the following conditions being satisfied: (1) item $j$ is not fully allocated before the last \texttt{normal\_update($i, j$)}; (2) some buyer leaves the temporary buyer set of some item during the last \texttt{normal\_update}; (3) the buyer $i'$ who loses item $j$ to buyer $i$ during the last \texttt{normal\_update($i, j$)} has no proposable item; (4) a circle is identified and a simultaneous update is conducted. Also note that the first three exit conditions can only be satisfied before a circle is identified.
	
	We will prove that, any one of the above  4 exit conditions will lead to one of the following events to happen: (a) the number of buyers with positive remaining budgets and non-empty set of proposable items strictly decreases (compared to the beginning of the current round of the main while loop); (b) some item becomes fully allocated for the first time during the current round; (c) some buyer leaves the temporary buyer set of some item during the current round.	
	\begin{enumerate}
		\item
		As discussed in the proof of correctness, if the current round of the main while loop does not finish after the last \texttt{normal\_update($i, j$)}, it must be the case where buyer $i$ runs out of its budget and only one buyer $i'$ loses item $j$ to buyer $i$.
		This indicates that the number of buyers with positive remaining budgets and non-empty set of proposable items will not increase after such a \texttt{normal\_update}.
		If item $j$ is not fully allocated before the last \texttt{normal\_update($i, j$)}, but becomes fully allocated afterwards, event (b) happens.
		Otherwise, buyer $i$ runs out of its budget and no one loses item $j$, thus event (a) happens.
		\item
		This is exactly event (c).
		\item
		The number of buyers with positive remaining budgets and non-empty set of proposable items will not increase before the last \texttt{normal\_update($i, j$)}.
		During the last update, buyer $i$ runs out of budget, and buyer $i'$ has no proposable item (even before the update), so event (a) happens.
		\item 
		Before the second \texttt{normal\_update($i_0, j_0$)}, the number of buyers with positive remaining budgets and non-empty set of proposable items will not increase.
		If $\Delta B_0 > 0$ and $r = \frac{B_0}{\Delta B_0}$, then after the simultaneous update, $i_0$ runs out of budget, $i_1, \dots, i_\ell$'s budgets remains zero, so event (a) happens.
		Otherwise, suppose $r = \frac{y_{i_k, j_{k - 1}}}{\Delta y_{i_k, j_{k - 1}}}$ for some $k$.
		Then $x_{i_k, j_{k - 1}} = y_{i_k, j_{k - 1}}  - r \Delta y_{i_k, j_{k - 1}} = 0$, i.e., $i_k$ leaves the temporary buyer set of item $j_{k - 1}$ and event (c) happens.
	\end{enumerate} 
	
	Event (b) can happen at most once for each item, and event (c) can happen at most once for each (buyer, item) pair.
	Regarding event (a), an item that is not proposable to a buyer will never return to be proposable, and a buyer's remaining budget can only increase if it leaves some item's temporary buyer set (event (c)).
	Therefore event (a) can only happen a polynomial number of times, as well.
\end{proof}

\section{Missing Results and Proofs in Section \ref{sec:non_personalized_pricing_game}}
\label{app:non_personalized_pricing_game}

\subsection{Proof of Theorem \ref{thm:one_sided_tractability} and Relation to Network Flow and Maximal Matching}
\label{app:network_flow_and_maximal_matching}

Theorem \ref{thm:one_sided_tractability} can be decomposed into the following two parts, which we will prove separately.

\begin{lemma} \label{lemma:edmonds_karp}
	Algorithm \ref{alg:stable_edmonds_karp} returns a stable allocation with maximum revenue in strongly poly-time.
\end{lemma}

\begin{lemma} \label{lemma:minimum_maximal_matching}
	It is APX-hard to compute the minimum revenue over all stable allocations.
\end{lemma}

\begin{algorithm}
	%	\small
	\caption{Revenue-optimal stable allocation with non-personalized prices}
	\label{alg:optimal_non_personalized}
	\LinesNumbered
	\DontPrintSemicolon
	\SetAlgoNoEnd
	
	\KwIn{valuations $v_{i, j}$, prices $p_{j}$, budgets $B_i$}
	
	First solve the optimal revenue $R$ over all budget-feasible allocations through the linear program:
	\begin{align*}
		\max \quad &
		\sum_{i, j} p_{j} x_{i, j}
		\\
		\text{s.t.} \quad &
		\sum_i x_{i, j} \leq 1, \forall j
		\\
		& \sum_j p_j x_{i, j} \leq B_i, \forall i
		\\
		& x_{i, j} = 0, \text{if $v_{i, j} = 0$, $\forall i, j$}
		\\
		& x_{i, j} \geq 0, \forall i, j.
	\end{align*}
	
	Then solve another linear program to get a stable allocation $x$ by replacing the objective with the market-wise welfare and adding the constraint that the revenue equals $R$.
	\begin{align*}
		\max \quad &
		\sum_{i, j} v_{i, j} x_{i, j}
		\\
		\text{s.t.} \quad &
		\sum_{i, j} p_j x_{i, j} \geq R
		\\
		& \sum_i x_{i, j} \leq 1, \forall j
		\\
		& \sum_j p_j x_{i, j} \leq B_i, \forall i
		\\
		& x_{i, j} = 0, \text{if $v_{i, j} = 0$, $\forall i, j$}
		\\
		& x_{i, j} \geq 0, \forall i, j.
	\end{align*}
	
\end{algorithm}

\begin{algorithm}
	%	\small
	\caption{Stable Edmonds-Karp algorithm}
	\label{alg:stable_edmonds_karp}
	\LinesNumbered
	\DontPrintSemicolon
	\SetAlgoNoEnd
	
	\KwIn{valuations $v_{i, j}$, prices $p_{j}$, budgets $B_i$, a compatible allocation $x$}
	
	construct the corresponding flow network $G = (V, E)$ \;
	
	construct the initial flow $f$ from the input compatible allocation $x$: $f_{s, i} = -f_{i, s} = \sum_j p_j x_{i, j}, \forall i \in I$, $f_{j, t} = - f_{t, j} = \sum_i p_j x_{i, j}, \forall j \in J$, and $f_{i, j} =  - f_{j, i} = p_j x_{i, j}, \forall i \in I, j \in J$
	
	\Repeat{there is no path from $s$ to $t$ in the graph $(V, E^f)$}{
		find a shortest path $p$ from $s$ to $t$ in the graph $(V, E^f)$ \;
		
		suppose that $p$ ends with edge $(i_0, j_0)$ and $(j_0, t)$ where $i_0 \in I$ and $j_0 \in J$ (this must be the case due to the bipartite structure of $G$) \;
		
		$j_1 \leftarrow \arg \min_{j: v_{i_0, j} > 0 \text{ and } \sum_i f_{i, j} < p_j} \frac{p_j}{v_{i, j}}$ (the set $\{j: v_{i_0, j} > 0 \text{ and } \sum_i f_{i, j} < p_j\}$ is not empty since $j_0$ lies in it; $(i_0, j_1), (j_1, t) \in E^f$ since $c_{i_0, j_1} - f_{i_0, j_1} \geq c_{j_1, t} - \sum_i f_{i, j_1} = p_j - \sum_i f_{i, j_1} > 0$)
		\label{alg:stable_edmonds_karp:most_preferred}
		
		augment $f$ along the augmenting path $p' = p \cup \{(i_0, j_1), (j_1, t)\} \setminus \{(i_0, j_0), (j_0, t)\}$ by $c^f(p') = \min\{c^f_{w_1, w_2}: (w_1, w_2) \in p'\}$
	}
\end{algorithm}

For Lemma \ref{lemma:edmonds_karp}, we start with Algorithm \ref{alg:optimal_non_personalized}, which is built on linear programming.

\begin{lemma} \label{lemma:stable_max_revenue_non_personalized_lp}
	Algorithm \ref{alg:optimal_non_personalized} returns a stable allocation with maximum revenue in (weakly) poly-time.
\end{lemma}
\begin{proof}
	Clearly $R$ is an upper bound on revenue for any budget-feasible allocation.
	We only need to show that the linear program at the right returns a stable allocation $x$.
	
	Suppose otherwise, then there exist a buyer $i_0$ and an item $j_0$ such that $v_{i, j} > 0, \sum_i x_{i, j_0} < 1$ and it holds that either (1) buyer $i_0$'s remaining budget is not zero or (2) buyer $i_0$ strictly prefers $j_0$ to some other item $j_1$ with $x_{i_0, j_1} > 0$.
	If case (1) holds, we can get a budget-feasible allocation with a revenue strictly larger than $R$, a contradiction.
	If case (2) holds, for buyer $i_0$, the bang-per-buck of item $j_0$ is strictly larger than item $j_1$. So we can modify $x$ by moving some budget of buyer $i_0$ spent on item $i_1$ to item $i_0$ and get a feasible solution of the linear program at the right with strictly larger objective value, again a contradiction.
\end{proof}

In the linear program given in Algorithm \ref{alg:optimal_non_personalized}, if we treat $p_j x_{i, j} =: f_{i, j}$ as the decision variables, the first linear program actually encodes a maximum flow problem: the network $G$ consists of $V = \{s, t\} \cup I \cup J$ as vertices, with source $s$ and sink  $t$; the capacity of edge $(s, i)$, denoted as $c_{s, i}$, is $B_i$; for every $j \in J$, $c_{i, j} = p_j$ for every buyer $i$ with $v_{i, j} > 0$, and $c_{j, t} = p_j$; other unmentioned edges have a capacity of zero.
Any budget-feasible allocation with the optimal revenue is equivalent to a maximum flow of the network. %(the converse is not true?)
Given a legitimate flow $f$, the residual network $G^f$ is identical to $G$ except for capacities, defined as $c^f_{w_1, w_2} = c_{w_1, w_2} - f_{w_1, w_2}$ for any $w_1, w_2 \in V$, and edges $E^f = \{(w_1, w_2): c^f_{w_1, w_2} > 0\}$.
%Denote the network corresponding to prices $p$ as $G(p)$.
We also adopt the conventional terminology of augmenting path and augmenting a flow along a path.

With this connection to network flow, we develop algorithm \ref{alg:stable_edmonds_karp} by adapting the classic Edmonds-Karp algorithm such that the compatibility of flows (allocations) could be maintained throughout the execution of the algorithm.

\begin{lemma}[Item-wise monotonicity] \label{lemma:edmonds_karp_monotonicity}
	During the execution of algorithm \ref{alg:stable_edmonds_karp}, $f_{j, t} = \sum_i f_{i, j}$ weakly increases.
\end{lemma}
\begin{proof}
	Due to the bipartite structure of $G^f$, every path $p = \{(s, i_0), (i_0, j_0), \dots, (i_\ell, j_{\ell}, (j_{\ell}, t))\}$ must alternate between $I$ and $J$ except for the beginning and the ending edge.
	For each $j_k$ $(k \neq \ell)$, $f_{j_k, t} = \sum_i f_{i, j_k}$ is not changed during the update, while $f_{j_\ell, t} = \sum_i f_{i, j_\ell}$ strictly increases.
\end{proof}

\begin{proof}[Proof of Lemma \ref{lemma:edmonds_karp}]
	We only need to show that compatibility of the allocation corresponding to $f$ is always maintained.
	Suppose by contradiction that the allocation corresponding to $f$ is compatible, but after one augmenting step along path $p$, the new allocation $x'$ becomes incompatible.
	Let $i_0$ be a buyer for whom $x'$ is incompatible, and $i_0$ prefer item $j_1$ (with $\sum_{i} x'_{i, j_1} < 1$) over item $j_0$ (with $x'_{i_0, j_0} > 0$).
	
	By lemma \ref{lemma:edmonds_karp_monotonicity}, $\sum_{i} x_{i, j_1} \leq \sum_{i} x'_{i, j_1} < 1$.
	Then by compatibility of $x$, $x_{i_0, j_0} = 0$ and thus $(i_0, j_0) \in p$.
	First, $(j_0, t)$ cannot be the last edge of $p$, for otherwise $i_0$ should have chosen an item it prefers to $j_0$ at line \ref{alg:stable_edmonds_karp:most_preferred}.
	However, $(i_0, j_1), (j_1, t) \in E^f$ since $j_1$ is not fully allocated yet.
	Hence we find a strictly shorter $s$-$t$ path by deviating from $p$ at $i_0$, a contradiction.
\end{proof}

Recall that maximum (cardinality) matching can be modeled as a maximum flow problem.
To prove Lemma \ref{lemma:minimum_maximal_matching}, we show that the revenue-minimizing stable allocation can be used to perfectly encode a minimum maximal matching.

\begin{proof}[Proof of Lemma \ref{lemma:minimum_maximal_matching}]
	We give a reduction to our problem from the minimum maximal matching problem, which is NP-complete even in the case of a bipartite graph.
	Given a bipartite graph $(I, J, E)$ with vertex sets $I, J$ and $E \subseteq I \times J$, a maximal matching $M$ is a subset of $E$ such that no two edges in $M$ share a vertex, and every edge in $E$ shares a vertex with at least one edge in $M$.
	The (decision) problem of minimum maximal matching is to ask, given a number $T$, whether there exists a maximal matching with size at most $T$.
	
	Given a bipartite graph $(I, J, E)$, construct a market with buyer set $I$ and item set $J$ (we will refer to $i \in I$ either as a vertex or a buyer, similarly for $j \in J$).
	Let $p_j = 1$ for all items and $B_i = 1$ for all buyers.
	If $(i, j) \in E$, let $v_{i, j} = 1$. Otherwise $v_{i, j} = 0$.
	
	Note that buyers have no preference over their interested items, 
	then it is straightforward to see that maximality of an integral matching is equivalent to stability of its corresponding integral allocation (in the rest of the proof we will use matching and allocation interchangeably).
	We only need to show that, given a stable allocation $x$ (possibly fractional) with a revenue of $T$, there exists an integral stable allocation with a revenue no larger than $T$.
	
	Given a stable allocation $x$, let $I' \subseteq I$ be the set of buyers depleting budgets in $x$, and $J' \subseteq J$ be the set of items fully allocated in $x$.
	Let $y$ be a maximum integral matching in the bipartite graph $(I', J', E')$ where $(i, j) \in E'$ if $v_{i, j} > 0$.
	Note that if we restrict $x$ to $(I', J')$, it is a legitimate fractional matching (flow).
	By the integral flow theorem, we have
	\begin{displaymath}
		\sum_{i \in I', j \in J'} x_{i, j} \leq \sum_{i \in I', j \in J'} y_{i, j}.
	\end{displaymath}
	Let $I''$ be the set of buyers depleting budgets in $y$, and $J''$ be the set of items fully allocated in $y$.
	Extend $y$ to construct an allocation $z$ by adding a maximum matching between $I \setminus I''$ and $J \setminus J''$.
	Since $y$ is a perfect matching between $I''$ and $J''$, $z$ must be a maximal (stable) matching.
	For the revenue of $z$, consider buyers in $I'$ and $I \setminus I'$: (1) buyers in $I'$ deplete budgets in $x$, so they could not spend more in $z$; (2) buyers in $I \setminus I'$ are only interested in items in $J'$ by stability of $x$, therefore they spend at most
	\begin{align*}
		& |J'| - \sum_{i \in I', j \in J'} z_{i, j} 
		= 
		|J'| - \sum_{i \in I', j \in J'} y_{i, j} 
		\leq 
		|J'| - \sum_{i \in I', j \in J'} x_{i, j}
		= 
		\sum_{i \in I \setminus I', j \in J} x_{i, j},
	\end{align*}
	which concludes the reduction.
	
	Since it is a strict reduction, we inherit APX-hardness of the minimum maximal matching problem for bipartite graphs \cite{chlebik2006approximation}.
\end{proof}

\begin{corollary} \label{corollary:np_hardness_seller_revenue}
	Given a subset $J' \subset J$, it is APX-hard to either minimizes or maximizes $\sum_{i \in I, j \in J'} p_j x_{i, j}$ over all the revenue-optimal stable allocations.
\end{corollary}
\begin{proof}
	We use the same reduction as in the previous proof, except that we add an item $j_0$ to the market with $p_{j_0} = 2|I|$ and $v_{i, j_0} = |I|$ for all $i \in I$.
	The total revenue of the market is always $|I|$ (since all buyers can deplete their remaining budgets on $j_0$), and since $j_0$ has a strictly lower bang-per-buck than items in $J$, a stable allocation in this market remains stable when the item set is restricted to $J$.
	The rest of the reduction is essentially the same.
	
	Hardness of the maximization problem can be deduced by taking the complement.
\end{proof}

%\subsection{Proof of Theorem \ref{thm:non_personalized_nash}}
%\label{app:non_personalized_nash}

\subsection{Budget Depletion Problem}
\label{app:budget_depletion}

It is not clear whether buyers will deplete their budgets at any maximin equilibrium.
Nonetheless, we show that, if some buyer does not deplete its budget, some seller could improve its \textit{worst case} revenue.
Recall that $R_k(p)$ is the minimum revenue received by seller $k$ over all (market-wise) revenue-optimal stable allocations with pricing $p$.
Let $R(p)$ be the market-wise optimal revenue over all stable allocations.

\begin{lemma} \label{lemma:linear_increase}
	If $p'$ differs from $p$ only for a single item $j$ with $p_j' = p_j + \epsilon$ for some $\epsilon > 0$, then $R(p') \leq R(p) + \epsilon$.
\end{lemma}
\begin{proof}
	Every feasible allocation with $p'$ is also a feasible allocation with $p$, and the objective changes by at most $\epsilon$.
\end{proof}

%\begin{lemma}
%	Given two budget-feasible allocation $x$ and $y$ (with possibly different price vectors), if the sets of fully allocated items and the sets of buyers whose budgets are depleted coincide for $x$ and $y$, then the stability of $x$ implies the stability of $y$.
%\end{lemma}
%\begin{proof}
%	Suppose by contradiction that $x$ is stable but $y$ is not.
%	Then there exists buyer $i_0$ in $y$ such that either (1) $i_0$'s budget is not depleted and there is item $j_0$ with $\sum_i y_{i, j_0} < 1$, or (2) $i_0$
%\end{proof}

\begin{proposition}
	If there exists some buyer $i_0$ such that $\sum_{j : v_{i_0, j} > 0} p_j < B_{i_0}$,
	then seller $k_0$ with $v_{i_0, j_0} > 0, j_0 \in J_{k_0}$ can strictly increase $R_{k_0}$ by increasing $p_{j_0}$ to $(1 + \epsilon) p_{j_0}$ for some $\epsilon > 0$.
\end{proposition}
\begin{proof}
	Let the new price vector be $p'$ and take $\epsilon$ such that $\sum_{j : v_{i_0, j} > 0} p'_j < B_{i_0}$.
	For any stable allocation $x'$ with $p'$, construct $x$ to be identical to $x'$ except for item $j_0$ with $x_{i, j_0} = \frac{x'_{i, j_0}}{1 + \epsilon}$ for all $i \neq i_0$ and $x_{i_0, j_0} = \frac{x_{i_0, j_0}' + \epsilon}{1 + \epsilon}$.
	$x$ is clearly budget-feasible and, by stability of $x'$, all items in which $i_0$ is interested are fully allocated in both $x'$ and $x$.
	
	Suppose that $x$ is not stable, then there exist a buyer $i_1$ ($\neq i_0$) and an item $j_1$ (with $v_{i_0, j_1} = 0$) such that $v_{i_1, j_1} > 0$ and $\sum_i x_{i, j_1} < 1$, and it holds that either (1) buyer $i_1$'s remaining budget is not zero or (2) buyer $i_1$ strictly prefers $j_1$ to some other item $j_2$ with $x_{i_1, j_2} > 0$.
	Since $x$ and $x'$ share the same set of buyers whose budgets are depleted and the same set of fully allocated items, if condition (1) holds for $x$, it also holds for $x'$ and leads to contradiction.
	If condition (2) holds, then $j_2 \neq j_0$ since the prices and allocations of $j_1$ and $j_2$ are identical in $x$ and $x'$.
	If $j_2 = j_0$, its price is increased in $p'$ so $i_1$ strictly prefers it less with $p'$ than with $p$.
	Hence $i_1$ also strictly prefers $j_0$ to $j_1$ in $p$, a contradiction.
	
	Through the above construction, for any stable allocation with $p'$, we can find a stable allocation with $p$ where the revenue of seller $k_0$ is smaller by $\epsilon$.
	Therefore by deviating from $p$ to $p'$, $R_{k_0}$ increases by at least $\epsilon$.
\end{proof}

\begin{lemma} \label{lemma:min_cut}
	There exists a minimum cut $(S, T)$ of $G(p)$, the flow network corresponding to prices $p$, that satisfies: (1) there is no cross edge between $S \cap I$ and $T \cap J$; (2) in any revenue optimal budget-feasible allocation, buyers in $T$ will always deplete their budgets on items in $T$, and items in $S$ will always be fully allocated among buyers in $S$.
\end{lemma}
\begin{proof}
	If there is an edge $(i, j)$ between $S \cap I$ and $T \cap J$, $j$ can be moved from $T$ to $S$ without increasing the total capacity (the capacity decreases by at least $c(i, j) = p_j$ and increases by at most $c(j, t) = p_j$).
	
	By max-flow min-cut theorem, in any revenue optimal budget-feasible allocation, buyers in $T$ will always deplete their budgets and items in $S$ will always be fully allocated.
	Note that $R(p) = \sum_{i \in T \cap I} B_i + \sum_{j \in S} p_j$.
	Buyers in $T$ spends $\sum_{i \in T \cap I} B_i$ in total, so buyers in $S$ must spend $\sum_{j \in S} p_j$.
	Since there is no cross edges between $S \cap I$ and $T \cap J$, buyers in $S$ can only spend money on items in $S$.
	So items in $S$ must be fully allocated among buyers in $S$, and thus buyers in $T$ can only spend money on items in $T$.
\end{proof}

\begin{proposition} \label{thm:budget_depletion}
	Given pricing $p$, if some buyer does not deplete its budget, there exists a profitable deviation for some seller and it can be computed in poly-time.
\end{proposition}

\begin{proof}[Proof of Proposition \ref{thm:budget_depletion}]
	Let $(S, T)$ be a minimum-cut of the flow network with prices $p$ without cross edges between $S \cap I$ and $T \cap J$.
	Suppose that buyer $i_0$ has a remaining budget of $\epsilon > 0$ at some revenue-optimal stable allocation $x$.
	Then $i_0 \in S$ (by lemma \ref{lemma:min_cut}) and all items that $i_0$ is interested in are also in $S$ (by the choice of minimum cut).
	Let $k_0$ be one of the sellers that own an item $j_0$ that is interested in by $i_0$.
	
	Consider a deviation of $k_0$ by changing $p_{j_0}$ to $p'_{j_0} = p_{j_0} + \epsilon / 2$.
	By lemma \ref{lemma:linear_increase}, $R(p') \leq R(p) + \epsilon / 2$.
	We can also construct a budget-feasible allocation with $p'$ from $x$ with revenue $R(p) + \epsilon / 2$ by letting $i_0$ pays more on $j_0$ by $\epsilon / 2$.
	Therefore $R(p') = R(p) + \epsilon / 2$.
	Note that the capacity of cut $(S, T)$ with prices $p'$ is exactly $R(p')$, so it is still a minimum cut.
	
	Consider any revenue-optimal stable allocation $y$ with $p'$.
	Then $y$ is a revenue-optimal allocation with $p$.
	For buyers in $S$, all their interested items are in also $S$, which are fully allocated. Hence $y$ is stable for them with $p$.
	For buyers in $T$, they only spend money on items in $T$, so their payments are identical in $p$ and $p'$.
	Hence $y$ is stable for them with $p$ as well.
	
	Since any revenue-optimal stable allocation $y$ with $p'$ is also revenue-optimal and stable with $p$ and seller $k_0$'s revenue in $y$ is increased by $\epsilon / 2$ from $p$ to $p'$, we have that $R_{k_0}(p') \geq R_{k_0}(p) + \epsilon / 2$.
\end{proof}

\section{Missing Proofs in Section \ref{sec:non_personalized_arbitrary}}
\label{app:non_personalized_arbitrary}

The proof of Theorem \ref{thm:apx_non_personalized_arbitrary} shares the same key construction of the proof of Theorem \ref{thm:apx_complete}.
See the proof and remarks in Appendix \ref{app:proof_apx_hard} on how the construction therein can be adapted to establish Theorem \ref{thm:apx_non_personalized_arbitrary}.

\section{Missing Results and Proofs in Section \ref{sec:personalized_pricing_game}}
\label{app:personalized_pricing_game}

The following lemmas and definitions will be used later.

\begin{lemma} \label{lemma:delete}
	Suppose that $x$ is a stable allocation w.r.t. pricing $p$, priorities $\geq$, valuations $v$, budgets $B$.
	Given some buyer $i$ and item $j$, let $B'$ be identical to $B$ except for $B'_i = B_i - p_{i, j} x_{i, j}$, and construct $(x', p', \geq', v')$ as follows:
	\begin{enumerate}
		\item if $x_{i, j} = 1$, construct $(x', p', \geq', v')$ by removing (entries of) item $j$ from $(x, p, \geq, v)$;
		\item if $x_{i, j} < 1$, construct $(x', p', \geq', v')$ to be identical to $(x, p, \geq, v)$ except for
		\begin{align*}
			x'_{i, j} & = 0, \\
			x'_{i', j} & = \frac{x_{i', j}}{1 - x_{i, j}}, \forall i' \neq i, \\
			p'_{i', j} & = (1 - x_{i, j}) p_{i', j}, \forall i' \text{ (including $i$)}, \\
			v'_{i', j} & = (1 - x_{i, j}) v_{i', j}, \forall i'  \text{ (including $i$)}; \\
		\end{align*}
	\end{enumerate}
	Then $x'$ is also stable allocation w.r.t. $(p', \geq', v', B')$.
\end{lemma}

\begin{definition}
	Given an allocation $x$, a \textbf{suballocation}  $y$ of $x$ is an allocation satisfying
	\begin{displaymath}
		y_{i, j} \leq x_{i, j}, \forall i, j.
	\end{displaymath}
\end{definition}

\subsection{Proof of Theorem \ref{thm:unique}}
\begin{proof}
	Suppose by contradiction that $x$ and $y$ are two distinct stable allocations.
	By lemma \ref{lemma:delete}, we can delete the common part of $x$ and $y$ (i.e., component-wise minimum of $x$ and $y$) from both to get two stable allocations in a new game.
	Therefore we can focus on $x$ and $y$ that satisfy $x_{i, j} y_{i, j} = 0$.
	
	Consider any buyer $i_0$ and item $j_0$ with $x_{i_0, j_0} > 0$.
	Then $y_{i_0, j_0} = 0$, which can only happen in two cases:
	\begin{enumerate}
		\item $\sum_{i} y_{i, j_0} = 1$ and $p_{i, j_0} > p_{i_0, j_0}$ for every item $j$ with $y_{i, j_0} > 0$.
		In this case, let $z^1$ be any suballocation of $y$ satisfying
		\begin{align*}
			\sum_{i} z^1_{i, j_0} & = x_{i_0, j_0}, \\
			\sum_i z^1_{i, j} & = 0, \forall j \neq j_0.
		\end{align*}
		\item $\sum_{j} p_{i_0, j} y_{i_0, j} = B_{i_0}$ and $j >_{i_0} j_0$ for every item $j$ with $y_{i_0, j} > 0$.
		Let $j$ be any item with $y_{i_0, j} > 0$, then $x_{i_0, j} = 0$ and $i_0$ strictly prefers $j$ to $j_0$.
		Recall that $i_0$ spends some money on $j_0$ in $x$, which means that $j$ is fully allocated in $x$ and for each buyer $i$ with $x_{i, j} > 0$, we have $p_{i, j} > p_{i_0, j}$.
		Let $z^1$ be any suballocation of $x$ satisfying
		\begin{align*}
			\sum_i z^1_{i, j} = y_{i_0, j}, \forall j.
		\end{align*}
	\end{enumerate}
	
	Depending on the situation, $z^1$ may be a suballocation of either $x$ or $y$.
	For the convenience of presentation, we rename $\{x, y\}$ to $\{x^0, y^0\}$ such that $z^1$ is a suballocation of $x^0$.
	In either case, if $z^1_{i, j} > 0$, a fraction of $y^0_{i_0, j}$ of item $j$ is available to $i$ in $y^0$.
	Let $w^1$ be the suballocation of $y^0$ where for each buyer $i$,
	\begin{align*}
		w^1_{i, j} & = y^0_{i, j}, \forall j, \text{ if there exists $j'$ such that $z^1_{i, j'} > 0$}, \\
		w^1_{i, j} & = 0,  \forall j, \text{ if for all $j'$, $z^1_{i, j'} = 0$}.
	\end{align*}
	In words, 
	$w^1$ is a suballocation of $y^0$ that includes the allocation of every buyer $i$ who spends some money in $z^1$.
	For all these buyers, there is some item that is available but not allocated to them in $y^0$, which implies that all their budgets are depleted.
	Therefore we have
	\begin{displaymath}
		\sum_{i, j} p_{i, j} w^1_{i, j} \geq \sum_{i, j} p_{i, j} z^1_{i, j}.
	\end{displaymath}
	
	Suppose $i$ and $j$ satisfy $z^1_{i, j} > 0$, then as analyzed above,
	%	then we have $y^0_{i, j} = 0$, $y^0_{i_0, j} > 0$ and $p_{i, j} > p_{i_0, j}$, i.e., $j$ is available to $i$ in $y^0$ but $i$ does not buy $j$.
	buyer $i$ spends all its budget on items that it strictly prefers to $j$ in $y^0$.
	However, in $x^0$, buyer $i$ spends no money on any item $j'$ with $y^0_{i, j'} > 0$, which implies that $j'$ is fully allocated in $x^0$ and $p_{i, j'} < p_{i', j'}$ for any buyer $i'$ with $x^0_{i', j'} > 0$.
	Therefore there exists $z^2$ that is a suballocation of $x^0$ satisfying
	\begin{align*}
		\sum_i z^2_{i, j} &= \sum_i w^1_{i, j}, \forall j,
	\end{align*}
	and we have
	\begin{align*}
		\sum_{i, j} p_{i, j} z^2_{i, j} 
		=  \sum_j \sum_{i: x^0_{i, j} > 0} p_{i, j} z^2_{i, j} 
		>   \sum_j \sum_{i: w^1_{i, j} > 0} p_{i, j} w^1_{i, j} 
		= 
		\sum_{i, j} p_{i, j} w^1_{i, j}.
	\end{align*}
	
	By construction, for each $i$ and $j$ with $z^2_{i, j} > 0$, there exists $i'$ such that $y^0_{i, j} = 0$, $y^0_{i', j} > 0$ and $p_{i, j} > p_{i', j}$ (for $z^1$, $i' = i_0$).
	So we can repeat the above analysis to get an infinite sequence $z^3, z^4, \dots$ of suballocations of $x^0$ with strictly increasing total payments.
	The payment increases at a rate lowered bounded by the strictly positive constant
	\begin{displaymath}
		\min_{i, i', j: p_{i, j} > p_{i', j}} \frac{p_{i, j}}{p_{i', j}},
	\end{displaymath}
	which is impossible since the total revenue must be finitely bounded.
	
	% wrong idea, but not sure whether useful or not
	%	Let $i_0$ be a buyer such that $x_{i_0} \neq y_{i_0}$, and $j_0$ be the item such that $j_0 \geq_{i_0} j$ for each item $j$ with $x_{i, j} \neq y_{i, j}$.
	%	WLOG, suppose $x_{i_0, j_0} > y_{i_0, j_0}$.
	%	
	%	First, we have $\sum_i y_{i, j_0} = 1$ and $i \geq_{j_0} i_0$ for each $i$ with $y_{i, j_0} > 0$, otherwise buyer $i_0$ could transfer some budget (either unspent or spent on an item $j$ with $j <_{i_0} j_0$) to $j_0$ to get a better bundle over available items in allocation $y$.
	%	Since $x_{i_0, j_0} > y_{i_0, j_0}$ and $\sum_i x_{i, j_0} \leq 1 = \sum_i y_{i, j_0}$, there exists a buyer $i_1$ with $x_{i_1, j_0} < y_{i_1, j_0}$.
	%	In allocation $x$, the fraction $x_{i_0, j_0}$ of item $j_0$ allocated to $i_0$ is available to $i_1$ ($i_1 >_{j_0} i_0$), but $i_1$ does not include it in its optimal bundle, which indicates its budget is depleted and $j \geq_{i_1} j_1$ for each item $j$ with $x_{i_1, j} > 0$.
	%	Since $x_{i_1, j_0} < y_{i_1, j_0}$ and $\sum_{j} p_{i_1, j} x_{i_1, j} = B_i \geq \sum_{j} p_{i_1, j} y_{i_1, j}$, there exists an item $j_1$ with $x_{i_1, j_1} > y_{i_1, j_1}$.
	%	
	%	Now we have $\sum_i y_{i, j_1} = 1$ and $i \geq_{j_1} i_1$ for each $i$ with $y_{i, j_1} > 0$, otherwise $i_1$ could transfer some budget from $j_0$ ($y_{i_1, j_0} > x_{i_1, j_0} \geq 0$) to $j_1$ so as to get a better bundle ($j_1 >_{i_1} j_0$) over available items in allocation $y$.
	%	At this point, by repeating the argument in the previous paragraph, we can find a sequence of buyer-item pairs $(i_2, j_2), (i_3, j_3), \dots$ satisfying
\end{proof}

\subsection{Proof of Theorem \ref{thm:apx_complete}}
\label{app:proof_apx_hard}

\begin{proof}
	We will give an L-reduction from the problem MAX-3SAT-3, a restricted version of the well-known 3SAT problem.
	An instance of 3SAT consists of $n$ variables $\{v_i\}$\footnote{Do not be confused with valuation $v_{i, j}$. We will not refer to valuation in this proof (we only need preferences, which can be easily constructed with suitably chosen valuations), so we stick to the traditional notation $v$ for variables and there should be no ambiguity.} and $m$ clause $\{c_j\}$ of the form $c_j = (l_1 \lor l_2 \lor l_3)$ where $l_k \in \{\pm v_i\}, k = 1, 2, 3$ is a (positive/negative) literal of some variable.
	The optimization problem MAX-3SAT is to find the assignment maximizing the number of satisfied clauses.
	MAX-3SAT-3 is a subset of MAX-3SAT where each variable appears at most 3 times.
	MAX-3SAT-3 is known to be APX-complete \cite{ausiello1999complexity}.
	
	\begin{table*}[t]
		\centering
		\begin{tabular}{c c cc cc cc}
			\toprule
			&&  \multicolumn{2}{c}{choice items} & \multicolumn{2}{c}{clause-literal items} & \multicolumn{2}{c}{buffer items} \\
			&& $g^{+v_i}$ & $g^{-v_i}$ & $g^{c_k(+v_i)}, k = 1, 2, 3$ & $g^{c_k(-v_i)}, k = 1, 2, 3$ &  $g_{\text{buffer}}^{+v_i}$ & $g_{\text{buffer}}^{-v_i}$ \\
			\midrule
			choice	 buyer &   $b^{v_i}$ & 13 & 13  \\
			literal buyer & $b^{+v_i}$ & 12 & & 2  & & 6   \\
			literal buyer &  $b^{-v_i}$ &  & 12 & & 2 & & 6 \\
			\bottomrule
		\end{tabular}
		\caption{Prices in the variable gadget dedicated to $v_i$.}
		\label{tab:variable_gadget_apx}
	\end{table*}
	
	For an arbitrary MAX-3SAT-3 instance, construct a market as follows.
	For each variable $v_i$, if it appears at least once, create a gadget corresponding to it as in Table \ref{tab:variable_gadget_apx} which involves 3 buyers and 10 items.
	Here $g^{c_k(l)} (k = 1, 2, 3, l = \pm v_i)$ is associated with the clauses where $l$ appears (the numbering is arbitrary but fixed).
	If $l$ appears in less than 3 clauses, simply associate the rest items to nothing (they are \textit{not} deleted).
	Note that this notation is solely used to name an item: $g^{c_{k_1}(l_1)}$ and $g^{c_{k_2}(l_2)}$ with $l_1 \neq l_2$ or $k_1 \neq k_2$ are two distinct items, even if they are both associated with the same clause.
	For each clause $c_j$, create a clause buyer $b^{c_j}$.
	If literal $l$ appears in $c_j$ and the item $g^{c_k(l)}$ is associated with $c_j$, the price of item $g^{c_k(l)}$ for buyer $b^{c_j}$ is set to $1$.
	Other prices are set to zero.
	A proper valuation can be chosen to satisfy: (1) each choice buyer is indifferent between two choice items; (2) each literal buyer strictly prefers the choice item and clause-literal items to buffer items; (3) clause buyers have no preferences; (4) buyers are not interested in items with a personalized price zero.
	Each choice buyer has a budget of $13$, each literal buyer has $12$ and each clause buyer has $1$.
	The total budget of all buyers is $37n + m$ (here $n$ is the number of variables appearing at least once).
	In this proof the allocation of item $j$ for buyer $i$ is denoted as $x(i, j)$.
	
	Given a stable allocation of the market, construct an assignment by setting to TRUE those literals whose literal buyer spends strictly more than 6 on the corresponding choice item.
	If two literal buyers both spend 6 on their choice items, set the variable to arbitrary value.
	The assignment is feasible since choice buyer $b^{v_i}$ will always spend all its budget on its choice items, and literal buyers $b^{+v_i}$ and $b^{-v_i}$ combined can only spend at most 12 on these choice items.
	
	Let OPT(A) and OPT(B) be the optimal objective value of the MAX-3SAT-3 instance and the optimal revenue of the constructed market respectively.
	Also suppose that the stable allocation has a total revenue $r$, and the constructed assignment satisfies $m'$ clauses.
	
	The maximum revenue of a stable allocation is at most the total budget $37n + m \leq 38n$.
	For the MAX-3SAT-3 instance, at least $m/2 \geq n/6$ clauses can be satisfied in the optimal solution.
	So $\text{OPT(B)} \leq (38 \times 6) \cdot  \text{OPT(A)}$.
	
	In any stable allocation, the budgets of all choice buyers and literal buyers are always depleted.
	To see this, simply note that, in any stable allocation, choice items are always available to their corresponding choice buyers, and clause-literal items and buffer items are always available to their literal buyers.
	In any stable allocation, if clause buyer $b^{c_j}$ spends a non-zero amount of money, we have $x(b^{c_j}, g^{c_k(l)}) > 0$ for some clause-literal item $g^{c_k(l)}$ that is associated with $c_j$.
	Then literal buyer $b^l$ must spend strictly more than 6 on choice item $g^l$ (otherwise it should include $g^l$ entirely in its optimal bundle), which means $l$ evaluates TRUE and so does $c_j$.
	Hence we have $r - 37n \leq m'$, and thus
	\begin{displaymath}
		\text{OPT(A)} - m' \leq 
		\text{OPT(B)} - r,
	\end{displaymath}
	which concludes the L-reduction.
\end{proof}

\paragraph{Remarks.}
The construction in the proof shows that the problem is APX-hard even for the restricted case where (1) item’s priorities are natural and strict among buyers with non-zero valuations; (2) items can be grouped into two sets, and every buyer is indifferent among items in the same group (in the reduction, choice and clause-literal items constitute one group, buffer items constitute the other). Also note that market-wise hardness result naturally transfers to the problem for each seller.

If we change the prices of all choice items to 12 for all buyers, the pricing of the market becomes non-personalized.
If the priorities of all choice items are set to prioritize their corresponding choice buyers the most, and the budget of each choice buyer is changed to 12, the proof proceeds almost exactly.
This shows that it is also APX-hard to find the maximum revenue over all stable allocations in non-personalized pricing markets with arbitrary priorities (the first half of Theorem \ref{thm:apx_non_personalized_arbitrary}).
By prioritizing the ``TRUE'' literal buyer and de-prioritize the ``FALSE'' one for each pair of choice items, any assignment of the MAX-3SAT-3 problem can be implemented precisely in the constructed market as the essentially unique stable allocation.
Therefore it is APX-hard to find the maximum revenue over all stable allocations and all priorities in non-personalized pricing markets (the second half of Theorem \ref{thm:apx_non_personalized_arbitrary}).

%\subsection{2-approximation}
%\label{app:2approximation}
%

%\subsection{Example Market without Nash Equilibrium}
%\label{app:no_pne}
%

\subsection{2-approximation}
\label{app:2approximation}

\begin{theorem} \label{thm:2_approximation}
	Let $x$ and $y$ be any two stable allocations. We have
	\begin{displaymath}
		\sum_{i, j} p_{i, j} x_{i, j} \leq 2 \sum_{i, j} p_{i, j} y_{i, j}.
	\end{displaymath}
\end{theorem}
\begin{proof}
	%	Similar to the proof of theorem \ref{thm:unique},
	By Lemma \ref{lemma:delete}, 
	we can delete the common part of $x$ and $y$ to get two stable allocations $x^0$ and $y^0$ in a new market where the corresponding items are removed.
	Note that $\sum_{i, j} p_{i, j} x^0_{i, j} \leq 2 \sum_{i, j} p_{i, j} y^0_{i, j}$ implies $\sum_{i, j} p_{i, j} x_{i, j} \leq 2 \sum_{i, j} p_{i, j} y_{i, j}$.
	Therefore we only need to focus on stable allocations $x$ and $y$ with the property that $x_{i, j} y_{i, j} = 0, \forall i, j$.
	
	Let $I_{x > y}$ be the set of buyers who spend strictly more in $x$ than in $y$, and $I_{x \leq y} = I \setminus I_{x > y}$.
	If $x_{i_0, j} > 0$ for some buyer $i_0 \in I_{x > y}$ and item $j$, we must have $\sum_i y_{i, j} = 1$, otherwise item $j$ is available for $i_0$ in $y$ but $i_0$ does not include it in its optimal bundle, which means $i_0$ depletes its budget in $y$, contradicting the fact that $i_0 \in I_{x > y}$.
	
	Due to the above property, there exists a suballocation $z$ of $y$ satisfying
	\begin{displaymath}
		\sum_{i} z_{i, j} = \sum_{i \in I_{x > y}} x_{i, j}, \forall j.
	\end{displaymath}
	Fix any item $j$, consider any buyer $i_1$ with $z_{i_1, j} > 0$ and any buyer $i_2 \in I_{x > y}$ with $x_{i_2, j} > 0$.
	If $p_{i_2, j} > p_{i_1, j}$, then item $j$ is available for $i_2$ in $y$ and $i_2$ does not include it in its optimal bundle.
	As argued above, this contradicts the fact that $i_2 \in I_{x > y}$.
	Hence we have the following inequality:
	\begin{align*}
		 \sum_{i, j} p_{i, j} z_{i, j} 
		= 
		\sum_{i} \sum_{j: z_{i, j} > 0} p_{i, j} z_{i, j} 
		\geq 
		\sum_{i \in I_{x > y}} \sum_{j: x_{i, j} > 0} p_{i, j} x_{i, j} 
		= 
		\sum_{i \in I_{x > y}, j} p_{i, j} x_{i, j}.
	\end{align*}
	
	On the other hand, by the definition of $I_{y \geq x}$,
	\begin{displaymath}
		\sum_{i \in I_{y \geq x}, j} p_{i, j} y_{i, j}
		\geq
		\sum_{i \in I_{y \geq x}, j} p_{i, j} x_{i, j}.
	\end{displaymath}
	
	Combining the above two inequalities gives the result we want.
\end{proof}

As a corollary, any algorithm to compute a stable allocation gives a 2-approximation to both the revenue-maximizing and minimizing problem.
The bound is tight as shown in the following example.

\begin{example}
	Consider a market with 2 buyers and 2 items, where $p_{1, 1} = p_{1, 2} = p_{2, 1} = 1$, $p_{2, 2} = 0$, $v_{1, 1} = v_{1, 2} = v_{2, 1} = 1$ and $v_{2, 2} = 0$.
	Then buyer 1 is indifferent between 2 items, but buyer 2 is only interested in item 1.
	Suppose both buyers have a budget of 1.
	Allocation $x$ with $x_{1, 1} = 1, x_{1, 2} = x_{2, 1} = x_{2, 2} = 0$ and $y$ with $y_{1, 2} = y_{2, 1} = 1, y_{1, 1} = y_{2, 2} = 0$ are both stable and satisfy
	\begin{displaymath}
		\sum_{i, j} p_{i, j} y_{i, j} = 2 \sum_{i, j} p_{i, j} x_{i, j}.
	\end{displaymath}
\end{example}

\subsection{Bang-per-buck Pricing}
\label{app:bang_per_buck_pricing}

In online advertising markets where the first price auction is used, if buyers (advertisers) are restricted to adopt auto-bidders using the multiplicative pacing strategy, they are effectively bang-per-buck priced \cite{conitzer2022pacing}.
This constitutes an example of bang-per-buck pricing in large real-world markets.
Theoretically, we have the following observation.

\begin{proposition}
	Suppose that buyer $i$ spends a non-zero amount of money $s$ on seller $k$ with an average bang-per-buck $r = \frac{1}{s} \sum_{j \in J_k} v_{i, j} x_{i, j}$ in a stable allocation $x$.
	We change the game a little bit by dividing each item $j \in J_k$ sold partially to buyer $i$ into two items, one sold fully to $i$ and one sold to other buyers according to $x_{i, j}$.
	Denote the set of items sold (now fully) to buyer $i$ as $J'_k \subseteq J_k$.
	
	If seller $k$ changes its pricing such that $p_{i, j} = \frac{v_{i, j}}{r}$ for all items in $J'_k$, and raises prices of items in $J_k \setminus J'_k$ high enough such that they are not available for buyer $i$, then buyer $i$ still gets an optimal bundle.
	If we also push prices of items in $J'_k$ high enough for buyers other than $i$, $x$ would still be stable.
\end{proposition}

The proposition can be proved using a standard exchange argument.
This is not a deviation analysis, and actually it is a rather weak observation.
But it at least provides some confidence for us to restrict each seller's pricing space to equalizing bang-per-buck ratios of all items for each buyer.
Note that bang-per-buck pricing cannot bring us the existence of pure Nash equilibrium: in Table \ref{tab:no_pne} each seller owns a single item (thus naturally bang-per-buck priced).

\subsection{Analysis of Example \ref{example:duopoly}}

\begin{proof}[Proof of Proposition \ref{proposition:duopoly_example_ce}]
	The most dangerous deviation for seller 1 is to lower $a_2$ (to any number) and raise $a_1$ to $p > 1$.
	If item 3 or item 4 is not fully allocated, buyer 1 will spend all money on seller 2, and seller 1 can only sell to buyer 2 to earn at most $2s + t < 2s + 2t$.
	This also implies item 3 will be fully sold to buyer 1.
	Item 1 and 2 must be fully allocated since they are the most preferred items of buyer 2 but the total price is no larger than 2.
	At stable allocation, buyer 1 and 2 will both deplete budgets, since the items for which they have top priorities have a total price no less than 2.
	
	As a result, the stable allocation must satisfy
	\begin{align*}
		2 t p x_{1, 2} + 2 (1 - t) p + (1 - s) p x_{1, 4}  & = 2
		\\
		2s + t(1 - x_{1, 2}) + 2(1 - s)(1 - x_{1, 4}) & = 2.
	\end{align*}
	And the revenue of seller 1 is
	\begin{displaymath}
		R_1 =
		2s + 4 - \bigparen{\frac{8}{3} - 2t} p - \frac{4}{3} \cdot \frac{1}{p}.
	\end{displaymath}
	With $t < \frac{2}{3}$, $R_1$ is maximized at $p = 1$.
	By symmetry, other combinations of deviation are not profitable, as well.
\end{proof}

\begin{proof}[Proof sketch of Proposition \ref{proposition:non_pareto_ne}]
	The proof is similar to the previous one by solving the system of equations regarding market clearance and budget depletion.
\end{proof}

We hypothesize the following result also holds:
if $s > 2/3$ and $t < 1/3$, the following bang-per-buck profile forms a Nash equilibrium:
\begin{align*}
	& a_2 = 
	\frac{2}{3} \cdot \frac{(1 + t)^2 + (2 - s)^2}{(2 - s)(1 + t)^2}
	, \\
	& a_1 =  b_2 = 1, \\
	& b_1 = 
	\frac{2}{3} \cdot \frac{(1 + t)^2 + (2 - s)^2}{(2 - s)^2 (1 + t)}.
\end{align*}
The numbers are derived from \textit{assuming} that the market clears and budgets are depleted, but we do not verify whether this presumption always holds.

%%% Use this environment to include acknowledgements (optional).
%%% This will be omitted in doubleblind mode.

%\begin{ack}
%By using the \texttt{ack} environment to insert your (optional) 
%acknowledgements, you can ensure that the text is suppressed whenever 
%you use the \texttt{doubleblind} option. In the final version, 
%acknowledgements may be included on the extra page intended for references.
%\end{ack}

%%% Use this command to include your bibliography file.

%\bibliography{bibliography}

\end{document}